\documentclass[pra,onecolumn]{revtex4}

\usepackage{amsmath, amssymb, amsthm, amsfonts, graphicx}
\usepackage[colorlinks]{hyperref}
\usepackage{complexity}
\usepackage{thm-restate}
\newcommand{\ket}[1]{| #1\rangle}        
\newcommand{\bra}[1]{\langle #1|}        
\newcommand{\braket}[2]{\langle #1 | #2 \rangle} 
\newcommand{\ketbra}[2]{| #1 \rangle\!\langle #2 |} 

\newtheorem{theorem}{Theorem}
\newtheorem{definition}[theorem]{Definition}
\newtheorem{corollary}[theorem]{Corollary}
\newtheorem{lemma}[theorem]{Lemma}

\newtheorem{threat}{Threat Model}

\newcommand{\eq}[1]{Eq.~\hyperref[eq:#1]{(\ref*{eq:#1})}}
\renewcommand{\sec}[1]{\hyperref[sec:#1]{Section~\ref*{sec:#1}}}
\newcommand{\app}[1]{\hyperref[app:#1]{Appendix~\ref*{app:#1}}}
\newcommand{\tab}[1]{\hyperref[tab:#1]{Table~\ref*{tab:#1}}}
\newcommand{\fig}[1]{\hyperref[fig:#1]{Figure~\ref*{fig:#1}}}
\newcommand{\figa}[2]{\hyperref[fig:#1]{Figure~\ref*{fig:#1}#2}}
\newcommand{\figx}[2]{\hyperref[fig:#1]{Figure~\ref*{fig:#1}(#2)}}
\newcommand{\thm}[1]{\hyperref[thm:#1]{Theorem~\ref*{thm:#1}}}
\newcommand{\lem}[1]{\hyperref[lem:#1]{Lemma~\ref*{lem:#1}}}
\newcommand{\cor}[1]{\hyperref[cor:#1]{Corollary~\ref*{cor:#1}}}
\newcommand{\defn}[1]{\hyperref[def:#1]{Definition~\ref*{def:#1}}}
\newcommand{\alg}[1]{\hyperref[alg:#1]{Algorithm~\ref*{alg:#1}}}
\newcommand{\prob}[1]{\hyperref[prob:#1]{Problem~\ref*{prob:#1}}}
\newcommand{\threatM}[1]{\hyperref[threat:#1]{Threat Model \ref*{threat:#1}}}

\input{qcircuit.tex}

\begin{document}
\title{Hardening Quantum Machine Learning Against Adversaries}
\author{Nathan Wiebe}
\affiliation{Quantum Architectures and Computing Group, Microsoft Research, Redmond, WA 98052}
\author{Ram Shankar Siva Kumar}
\affiliation{Microsoft, Redmond, WA 98052}
\begin{abstract}
Security for machine learning has begun to become a serious issue for present day applications.  An important question remaining is whether emerging quantum technologies will help or hinder the security of machine learning.  Here we discuss a number of ways that quantum information can be used to help make quantum classifiers more secure or private.  In particular, we demonstrate a form of robust principal component analysis that, under some circumstances, can provide an exponential speedup relative to robust methods used at present.  To demonstrate this approach we introduce a linear combinations of unitaries Hamiltonian simulation method that we show functions when given an imprecise Hamiltonian oracle, which may be of independent interest.  We also introduce a new quantum approach for bagging and boosting that can use quantum superposition over the classifiers or splits of the training set to aggragate over many more models than would be possible classically.  Finally, we provide a private form of $k$--means clustering that can be used to prevent an all powerful adversary from learning more than a small fraction of a bit from any user.  These examples show the role that quantum technologies can play in the security of ML and vice versa.  This illustrates that quantum computing can provide useful advantages to machine learning apart from speedups.  
\end{abstract}
\date{\today}
\maketitle

\section{Introduction}
There is a huge uptick in the use of machine learning for mission critical industrial applications – from self driving cars~\cite{pomerleau1989alvinn}  to detecting malignant tumors~\cite{cruz2006applications} to detecting fraudulent credit card transactions~\cite{xiao2015feature}: machine learning is crucial in decision making. However, classical Machine learning algorithms such as Principal component analysis (used heavily in anomaly detection scenarios), clustering (used in unsupervised learning), support vector machines (used in classification scenarios), as commonly implemented, are extremely vulnerable to changes to the input data, features and the final model parameters/hyper-parameters that have been learned. Essentially, an attacker can exploit any of the above vulnerabilities and subvert ML algorithms. As a result, an attacker has a variety of goals that can be achieved: increasing the false negative rate and thus become undetected (for instance, in the case of spam, junk emails are classified as normal)~\cite{wittel2004attacking} or by increasing the false positive rate (for instance, in the case of intrusion detection systems, attacks get drowned in sea of “noise” which causes the system to shift the baseline activity)~\cite{lowd2005good,stern2004linguistics,alfeld2016data,stern2004linguistics}, steal the underlying model itself exploiting membership queries~\cite{tramer2016stealing} and even recover the underlying training data breaching privacy contracts~\cite{tramer2016stealing}. 

In machine learning literature, this study is referred to as adversarial machine learning and has largely been applied to security sensitive areas such as intrusion detection~\cite{biggio2010multiple,biggio2013evasion} and spam filtering~\cite{wittel2004attacking}.  To combat the problem of adversarial machine learning, solutions have been studied from different vantage points: from a statistical stand point, adversaries are treated as noise and thus the models are hardened using robust statistics to overcome the malicious outliers~\cite{hampel2011robust}. Adversarial training is a  promising trend in this space, is wherein defenders train the system with adversarial examples from the start so that the model is acclimatized to such threats. From a security standpoint, there has been substantial work surrounding threat modeling machine learning systems~\cite{ringberg2007sensitivity,papernot2016limitations}  and frameworks for anticipating different kinds of attacks~\cite{papernot2016cleverhans}. 

Quantum computing has experienced a similar surge of interest of late and this had led to a synergistic relationship wherein quantum computers have been found to have profound implications for machine learning~\cite{biamonte2017quantum,aimeur2006machine,lloyd2013quantum,pudenz2013quantum,rebentrost2014quantum,WKS15,wiebe2016quantum,rebentrost2016quantum,
kerenidis2016quantum,amin2016quantum,schuld2017quantum,gilyen2017optimizing,kieferova2016tomography} and machine learning has been shown to be invaluable for characterizing, controlling and error correcting such quantum computers~\cite{assion1998control,granade2012robust,torlai2017neural}.  However, as of late the question of what quantum computers can do to protect machine learning from adversaries is relatively underdeveloped.  This is perhaps surprising given the impact that applications of quantum computing to security have a long history.

The typical nexus of quantum computing and security is studied from the light of quantum cryptography and its ramifications in key exchange and management. This paper takes a different tack.  We explores this intersection by asking questions the use of quantum subroutines to analyze patterns in data, and explores the question and security assurance - specifically, asking the question “Are quantum machine learning systems secure”?  

Our aim is to address this question by investigating the role that quantum computers may have in protecting machine learning models from attackers.  In particular, we introduce a robust form of quantum principal component analysis that provides exponential speedups while making the learning protocol much less sensitive to noise introduced by an adversary into the training set.  We then discuss bagging and boosting, which are popular methods for making models harder to extract by adversaries and also serve to make better classifiers by combining the predictions of weak classifiers.  Finally, we discuss how to use ideas originally introduced to secure quantum money to boost the privacy of $k$-means by obfuscating the private data of participants from even an allpowerful adversary.  From this we conclude that quantum methods can be used to have an impact on security and that while we have shown defences against some classes of attacks, more work is needed before we can claim to have a fully developed understanding of the breadth and width of adversarial quantum machine learning.

\section{Adversarial Quantum Machine Learning}
As machine learning becomes more ubiquitous in applications so too do attacks on the learning algorithms that they are based on.  The key assumption usually made in machine learning is that the training data is independent of the model and the training process.  For tasks such as classification of images from imagenet such assumptions are reasonable because the user has complete control of the data.  For other applications, such as developing spam filters or intrusion detection this may not be reasonable because in such cases training data is provided in real time to the classifier and the agents that provide the information are likely to be aware of the fact that their actions are being used to inform a model. 

Perhaps one of the most notable examples of this is the Tay chat bot incident.  Tay was a chat bot designed to learn from users that it could freely interact with in a public chat room.  Since the bot was programmed to learn from human interactions it could be subjected to what is known as a wolf pack attack, wherein a group of users in the chat room collaborated to purposefully change Tay's speech patterns to become increasingly offensive.  After $16$ hours the bot was pulled from the chat room.  This incident underscores the need to build models that can learn while at the same time resist malfeasant interactions on the part of a small fraction of users of a system.

A major aim of adversarial machine learning is to characterize and address such problems by making classifiers more robust to such attacks or by giving better tools for identifying when such an attack is taking place.  There are several broad classes of attacks that can be considered, but perhaps the two most significant in the taxonomy in attacks against classifiers are exploratory attacks and causative attacks.  Exploratory attacks are not designed to explicitly impact the classifier but instead are intended to give an adversary information about the classifier.  Such attacks work by the adversary feeding test examples to the classifier and then inferring information about it from the classes (or meta data) returned.  The simplest such attack is known as an evasion attack, which aims to find test vectors that when fed into a classifier get misclassified in a way that benefits the adversary (such as spam being misclassified as ordinary email).
More sophisticated exploratory attacks may even try identify to the model used to assign the class labels or in extreme cases may even try identify the training set for the classifier.  Such attacks can be deployed as a precursor to causitive attacks or can simply be used to violate privacy assumptions that the users that supplied data to train the classifier may have had.

Causitive attacks are more akin to the Tay example.  The goal of such attacks is to change the model by providing it with training examples.  Again there is a broad taxonomy of causitive attacks but one of the most ubiquitous attacks is the poisoning attack.  A poisoning attack seeks to control a classifier by introducing malicious training data into the training set so that the adversary can force an incorrect classification for a subset of test vectors.  One particularly pernicious type of attack is the boiling frog attack, wherein the amount of malicious data introduced in the training set is slowly increased over time.  This makes it much harder to identify whether the model has been compromised by users because the impact of the attack, although substantial, is hard to notice on the timescale of days.

While a laundry list of attacks are known against machine learning systems the defences that have been developed thus far are somewhat limited.  A commonly used tactic is to replace classifiers with robust versions of the same classifier.  For example, consider $k$--means classification.  Such classifiers are not necessarily robust because the intra-cluster variance is used to decide the quality of a clustering.  This means that an adversary can provide a small number of training vectors with large norm that can radically impact the cluster assignment.  On the other hand, if robust statistics such as the median, are used then the impact that the poisoned data can have on the cluster centroids is minimal.  Similarly, bagging can also be used to address these problems by replacing a singular classifier that is trained on the entire training set with an ensemble of classifiers that are trained on subsets the training set.  By aggragating over the class labels returned by this process we can similarly limit the impact of an adversary that controls a small fraction of the data set.

Given such examples, two natural questions arise: 1) ``How should we model threats in a quantum domain?'' and 2)  ``Can quantum computing be used to make machine learning protocols more secure?''.  In order to address the former problem it is important to consider the access model used for the training data for the quantum ML algorithm.  Perhaps the simplest model to consider is a QRAM wherein the training data is accessed using a binary access tree that allows, in depth $O(n)$ but size $O(2^n)$, the training data to be accessed as bit strings.  For example, up to isometries, $U_{\rm QRAM} \ket{j} \ket{0}= \ket{j}\ket{[v_j]}$ where $[v_j]$ is a qubit string used to encode the $j^{\rm th}$ training vector.  Similarly such an oracle can be used to implement a second type of oracle which outputs the vector as a quantum state vector (we address the issue of non-unit norm training examples below): $U \ket{j} \ket{0} = \ket{j}\ket{v_j}$.  Alternatively, one can consider a density matrix query for use in an LMR algorithm for density matrix exponentiation~\cite{lloyd2014quantum,kimmel2017hamiltonian} wherein a query to the training data takes the form $B\ket{0} \mapsto \rho$ where $\rho$ is a density operator that is equivalent the distribution over training vectors.  For example, if we had a uniform distribution over training vectors in a training set of $N$ training vectors then $\rho = \sum_{j=1}^N \ketbra{v_j}{v_j} / ({\rm Tr \sum_{j=1}^N \ketbra{v_j}{v_j}})$.

With such quantum access models defined we can now consider what it means to perform a quantum poisoning attack.  A poisoning attack involves an adversary purposefully altering a portion of the training data in the classical case so in the quantum case it is natural to define a poisoning attack similarly.  In this spirit, a quantum poisoning attack takes some fraction of the training vectors, $\eta$, and replaces them with training vectors of their choosing.  That is to say, if without loss of generality an adversary replaces the first $\eta N$ training vectors in the data set then the new oracle that the algorithm is provided is of the form
\begin{equation}
U_{\rm QRAM} \ket{j} \ket{0} = \begin{cases} \ket{j}\ket{[v_j]} &{\mbox{, if } j>\eta N}\\ \ket{j}\ket{[\phi_j]} &\mbox{, otherwise}  \end{cases},
\end{equation}
where $\ket{[\phi_j]}$ is a bit string of the adversary's choice.  Such attacks are reasonable if a QRAM is stored on an untrusted machine that the adversary has partial access to, or alternatively if the data provided to the QRAM is partially under the adversary's control.  The case where the queries provide training vectors, rather than qubit vectors, is exactly the same.  The case of a poisoning attack on density matrix input is much more subtle, however in such a case we define the poisoning attack to take the form $B\ket{0} = \rho':~|{\rm Tr}(\rho-\rho')|\le \eta$, however other definitions are possible.

Such an example is a quantum causitive attack since it seeks to cause a change in the quantum classifier.  An example of a quantum exploratory attack could be an attack that tries to identify the training data or the model used in the classification process.  For example, consider a quantum nearest neighbor classifier.  Such classifiers search through a large database of training examples in order to find the closest training example to a test vector that is input.  By repeatedly querying the classifier using quantum training examples adversarially chosen it is possible to find the decision boundaries between the classes and from this even the raw training data can, to some extent, be extracted.  Alternatively one could consider cases where an adversary has access to the network on which a quantum learning protocol is being carried out and seeks to learn compromising information about data that is being provided for the training process, which may be anonymized at a later stage in the algorithm.

Quantum can help, to some extent, both of these issues.  Poisoning attacks can be addressed by building robust classifiers.  That is, classifiers that are insensitive to changes in individual vectors.  We show that quantum technologies can help with this by illustrating how quantum principal component analysis and quantum bootstrap aggragation can be used to make the decisions more robust to poisoning attacks by proposing new variants of these algorithms that are ammenable to fast quantum algorithms for medians to be inserted in the place of the expectation values typically used in such cases.  We also illustrate how ideas from quantum communication can be used to thwart exploratory attacks by proposing a private version of $k$--means clustering that allows the protocol to be performed without (substantially) compromising the private data of any participants and also without requiring that the individual running the experiment be authenticated.

\section{Robust Quantum PCA}
The idea behind principal component analysis is simple.  Imagine you have a training set composed of a large number of vectors that live in a high dimensional space.  However, often even high dimensional data can have an effective low-dimensional approximation.  Finding such representations in general is a fine art, but quantum principal component analysis provides a prescriptive way of doing this.  The idea of principal component analysis is to examine the eigenvectors of the covariance matrix for the data set.  These eigenvectors are the principal components, which give the directions of greatest and least variance, and their eigenvalues give the magnitude of the variation.  By transforming the feature space to this eigenbasis and then projecting out the components with small eigenvalue one can reduce the dimensionality of the feature space.

One common use for this, outside of feature compression, is to detect anomalies in training sets for supervised machine learning algorithms.  Imagine that you are the administrator of a network and you wish to determine whether your network has been compromised by an attacker.  One way to detect an intrusion is to look at the packets moving through the network and use principal component analysis to determine whether the traffic patterns are anomalous based on previous logs.  The detection of an anomalous result can be performed automatically by projecting the traffic patterns onto the principal components of the traffic data.  If the data is consistent with this data it should have high overlap with the eigenvectors with large eigenvalue, and if it is not then it should have high overlap with the small eigenvectors.

While this technique can be used in practice it can have a fatal flaw.  The flaw is that an adversary can inject spikes of usage in directions that align with particular principal components of the classifier.  This allows them to increase the variance of the traffic along principal components that can be used to detect their subsequent nefarious actions.  
To see this, let us restrict our attention to poisoning attacks wherein an adversary controls some constant fraction of the training vectors.

Robust statistics can be used to help mitigate such attacks.  The way that it can be used to help make PCA secure is by replacing the mean with a statistic like the median.  Because the median is insensitive to rare but intense events, an adversary needs to control much more of the traffic flowing through the system in order to fool a classifier built to detect them.  For this reason, switching to robust PCA is a widely used tactic to making principal component analysis more secure.
To formalize this, let us define what the robust PCA matrix is first.
\begin{definition}
Let $x_j: j=1:N$ be a set of real vectors and let $e_k:k=1,\ldots,N_v$ be basis vectors in the computational basis then we define $\mathbf{M}_{k,\ell} := {\rm median}([e_k^Tx_j -{\rm median}(e_k^T x_j)][e_\ell^Tx_j -{\rm median}(e_\ell^T x_j)])$.\label{def:N_v}
\end{definition}
This is very similar to the PCA matrix when we can express as ${\rm mean}([e_k^Tx_j -{\rm mean}(e_k^T x_j)][e_\ell^Tx_j -{\rm mean}(e_\ell^T x_j)])$.  Because of the similarity that it has with the PCA matrix, switching from a standard PCA algorithm to a robust PCA algorithm in classical classifiers is a straight forward substitution.

In quantum algorithms, building such robustness into the classifier is anything but easy.  The challenge we face in doing so is that the standard approach to this uses the fact that quantum mechanical state operators can be viewed as a covariance matrix for a data set~\cite{lloyd2014quantum,kimmel2017hamiltonian}.  By using this similarity in conjunction with the quantum phase estimation algorithm the method provides an expedient way to project onto the eigenvectors of the density operator and in turn the principal components.  However, the robust PCA matrix given in~\ref{def:N_v} does not have such a natural quantum analogue.  Moreover, the fact that quantum mechanics is inherently linear makes it more challenging to apply an operation like the median which is not a linear function of its inputs.  This means that if we want to use quantum PCA in an environment where an adversary controls a part of the data set, we will need to rethink our approach to quantum principal component analysis.

The challenges faced by directly applying density matrix exponentiation also suggests that we should consider quantum PCA within a different cost model than that usually used in quantum principal component analysis.  We wish to examine whether data obtained from a given source, which we model as a black box oracle,
is typical of a well understood training set or not.  We are not interested in the principal components of the vector, but rather are interested in its projection onto the low variance subspace of the training data.  We also assume that the training vectors are accessed in an oracular fashion and that nothing apriori is known about them.

\begin{definition}
Let $U$ be a self-adjoint unitary operator acting on a finite dimensional Hilbert space such that $U\ket{j}\ket{0} = \ket{j}\ket{v_j}$ where $v_j$ is the $j^{\rm th}$ training vector.
\end{definition}

\begin{lemma}
Let $x_j\in \mathbb{C}^N$ for integer $j$ in $[1,N_v]$ obey $\|x_j\| \le R\in \mathbb{R}$.   There exists a set of unit vectors in a Hilbert space of dimension $N(2N_v+1)$ such that for any $\ket{x_j},|{x_k^\dagger}\rangle$ within this set $\braket{x_j^\dagger}{x_k} = \langle x_j,x_k\rangle/R^2$.\label{lem:isometry}
\end{lemma}
\begin{proof}
The proof is constructive.  Let us assume that $R>0$.  First let us define for any $x_j$,
\begin{equation}
\ket{x_j/\|x_j\|}:= x_j/\|x_j\|.
\end{equation}
We then encode
\begin{align}
&x_j \mapsto \ket{x_j} := \ket{x_j/\|x_j\|} \left(\frac{\|x_j\|}{R}\ket{0} + \sqrt{1-\frac{\|x_j\|^2}{R^2}}\ket{j} \right).\nonumber\\
&x^\dagger_j \mapsto |x_j^\dagger \rangle := \ket{x_j^\dagger/\|x_j\|} \left(\frac{\|x_j\|}{R}\ket{0} + \sqrt{1-\frac{\|x_j\|^2}{R^2}}\ket{j+N_v} \right).
\end{align}
If $R=0$ then
\begin{align}
&x_j \mapsto \ket{0} \ket{j},\nonumber\\
&x_j^\dagger \mapsto \ket{0} \ket{j+N_v}.
\end{align}
It is then easy to verify that $\braket{x_j^\dagger}{x_k} = \langle x_j,x_k\rangle/R^2$ for all pairs of $j$ and $k$, even if $R=0$.  The  resultant vectors are defined on  a tensor product of two vector spaces.  The first is of dimension $N$ and the second is or dimension at least $2N_v+1$.  Since the dimension of a tensor product of Hilbert spaces is the product of the subsystem dimensions, the dimension of the overall Hilbert space is $N(2N_v+1)$ as claimed.
\end{proof}
For the above reason, we can assume without loss of generality that all training vectors are unit vectors.  Also for simplicity we will assume that $v_j\in \mathbb{R}$, but the complex case follows similarly and only differs in that both the real and imaginary components of the inner products need to be computed.

\begin{threat}
Assume that the user is provided an oracle, $U$, that allows the user to access real-valued vectors of dimension as quantum states yielded by an oracle of the form in \lem{isometry}.  This oracle could represent an efficient quantum subroutine or it could represent a QRAM.  The adversary is assumed to have control with probability $\alpha <1/2$ the vector $\ket{x_j}$ yielded by the algorithm for a given coherent query subject to the constraint that $\|x_j\|\le R$.  The aim of the adversary is to affect, through these poisoned samples, the principal components yielded from the data set yielded by $U$. \label{threat:1}
\end{threat}

Within this threat model, the maximum impact that the adversary could have on any of the expectation values that form the principal component matrix $O(\alpha R)$.  Thus any given component of the principal component matrix can be controlled by the adversary provided that the data expectation $\alpha \gg \mu/R$ where $\mu$ is an upper bound on the componentwise means of the data set.  In other words, if $\alpha$ is sufficiently large and the maximum vector length allowed within the algorithm is also large then the PCA matrix can be compromised.  This vulnerability comes from the fact that the mean can be dominated by inserting a small number of vectors with large norm.

Switching to the median can dramatically help with this problem as shown in the following lemma, whose proof is trivial but we present for completeness.
\begin{lemma}
Let $P(x)$ be a probability distribution on $\mathbb{R}$ with invertable cummulative distribution function ${\rm CDF}$.  Further, let ${\rm CDF}^{-1}$ be its inverse cummulative distribution function and let ${\rm CDF}^{-1}$ be Lipshitz continuous with Lipshitz constant $L$ on $[0,1]$.  Assume that an adversary replaces $P(x)$ with a distribution $Q(x)$ that is within total variational distance $\alpha<1/2$ from $P(x)$.  Under these assumptions $|{\rm median}(P(x)) - {\rm median}(Q(x))|\le \alpha L$.\label{lem:perturb}
\end{lemma}
\begin{proof}
Let $\int_{-\infty}^{y_0} P(x) \mathrm{d}x =1/2$ and $\int_{-\infty}^{y_0} Q(x) \mathrm{d}x =1/2$.  We then have from the triangle inequality that
\begin{equation}
\int_{-\infty}^{y_1} Q(x) \mathrm{d}x \le \int_{-\infty}^{y_1} P(x) \mathrm{d}x + \int_{-\infty}^{\infty} |P(x)-Q(x)| \mathrm{d}x= {\rm CDF}(y_1) + \alpha.
\end{equation}
Thus we have that $\frac{1}{2} -\alpha \le {\rm CDF}(y_1)$.  The inverse cummulative distribution function is a monotonically increasing function on $[0,1]$ which implies along with our assumption of Lipshitz continuity with constant $L$
\begin{align}
y_0-L\alpha = {\rm CDF}^{-1}\left(\frac{1}{2}\right) - L\alpha &\le {\rm CDF}^{-1} \left(\frac{1}{2} - \alpha\right) \le {\rm CDF}^{-1}({\rm CDF}(y_1)) = y_1.
\end{align}
Using the exact same argument but applying the reverse triangle inequality in the place of the triangle inequality gives $y_1 \le y_0 +\alpha L$ which completes the proof.
\end{proof}

Thus under reasonable continuity assumptions on the unperturbed probability distribution~\lem{perturb} shows that the maximum impact that an adversary can have on the median of a distribution is negligible.  In contrast, the mean does not enjoy such stability and as such using robust statistics like the median can help limit the impact of adversarially placed training vectors within a QRAM or alternatively make the classes less sensitive to outliers that might ordinarily be present within the training data.

\begin{restatable}{theorem}{finalPCA}
Let the minimum eigenvalue gap of $d$--sparse Hermitian matrix $\mathbf{M}$ within the support of the input state be $\lambda$, let each training vector ${x}_j$ be a unit vector provided by the oracle defined in~\lem{isometry} and assume that the underlying data distribution for the components of the training data has a continuous inverse cummulative distribution with function with constant Lipshitz constant $L>0$. We have that 
\begin{enumerate}
\item The number of queries needed to the oracle given in~\lem{isometry} needed to sample from a distribution $P$ over $O(\lambda)$--approximate eigenvalues of $\mathbf{M}$ such that for all training vectors $x$ $|P(E'_n) - |\braket{x}{E_n}|^2| \le \Lambda\le 1$ is in $O([d^4/\Lambda^3 \lambda^3]\log^2(d^2/(\lambda \Lambda)))$.
\item Assuming an adversary replaces the data distribution used for the PCA task by one within variational distance $\alpha$ from the original such that $\alpha L \le 1$ and $\alpha < 1/2$.  If $\mathbf{M'}$ be is the poisoned robust PCA matrix then $\|\mathbf{M} -\mathbf{M'}\|_2 \le 5\alpha L(d+2)$.
\end{enumerate}

\label{thm:finalPCA}
\end{restatable}

The proof is somewhat involved.  It involves the use of linear-combinations of unitary simulation techniques and introduces a generalizations to these methods that shows that they can continue to function when probabilistic oracles are used for the matrix elements of the matrix $\mathbf{M}$ (see~\lem{Mbd}) without entanglement to the control register creating problems with the algorithm.  Furthermore, we need to show that the errors in simulation do not have a major impact on the output probability distributions of the samples drawn from the eigenvalues/ vectors of the matrix.  We do this using an argument based on perturbation theory under the assumption that all relevant eigenvalues have multiplicity $1$.  For these reasons, we encourage the interested reader to look at the appendix for all technical details regarding the proof.

From~\thm{finalPCA} we see that the query complexity required to sample from the eigenvectors of the robust PCA matrix is exponentially lower in some cases than what would be expected from classical methods.  While this opens up the possibility of profound quantum advantages for robust principal component analysis, a number of caveats exist that limit the applicability of this method:
\begin{enumerate}
\item The cost of preparing the relevant input data $\ket{x}$ will often be exponential unless $\ket{x}$ takes a special form.
\item The proofs we provide here require that the gaps in the relevant portion of the eigenvalue spectrum are large in order to ensure that the perturbations in the eigenvalues of the matrix do not sufficiently large as to distort the support of the input test vector.
\item The matrix $\mathbf{M}$ must be polynomially sparse.
\item The desired precision for the eigenvalues of the robust PCA matrix is inverse polynomial in the problem size.
\item The eigenvalue gaps are polynomially large in the problem size.
\end{enumerate}
These assumptions preclude it from providing exponential advantages in many cases, however, it does not preclude exponential advantages in quantum settings where the input is given by an efficient quantum procedure.  Some of these caveats can be relaxed by using alternative simulation methods or exploiting stronger promises about the structure of $\mathbf{M}$.  We leave a detailed discussion of this for future work.  The final two assumptions are in practice may be the strongest restrictions for our method, or at least the analysis we provide for the method, because the complexity diverges rapidly with both quantities.  It is our hope that subsequent work will improve upon this, perhaps by using more advanced methods based on LCU techniques to eschew the use of amplitude estimation as an intermediate step.

Looking at the problem from a broader perspective, we have shown that quantum PCA can be applied to defend against~\threatM{1}.  Specifically, we have resilience to . in a manner that helps defend against attacks on the training corpus either directly which allows the resulting classifiers to be made robust to adversaries who control a small fraction of the training data.  This shows that ideas from adversarial quantum machine learning can be imported into quantum classifiers.  In contrast, before this work it was not clear how to do this because existing quantum PCA methods cannot be easily adapted from using a mean to a median.  Additionally, this robustness also can be valuable in non-adversarial settigns becuase it makes estimates yielded by PCA less sensitive to outliers which may not necessarily be added by an adversary.  We will see this theme repeated in the following section where we discuss how to perform quantum bagging and boosting.

\section{Bootstrap Aggragation and Boosting}
Bootstrap aggragation, otherwise known as bagging for short, is another approach that is commonly used to increase the security of machine learning as well as improve the quality of the decision boundaries of the classifier.  The idea behind bagging is to replace a single classifier with an ensemble of classifiers.  This ensemble of classifiers is constructed by randomly choosing portions of the data set to feed to each classifier.  A common way of doing this is bootstrap aggragation, wherein the training vectors are chosen randomly with replacement from the original training set.  Since each vector is chosen with replacement, with high probability each training set will exclude a constant fraction of the training examples.  This can make the resultant classifiers more robust to outliers and also make it more challenging for an adversary to steal the model.

We can abstract this process by instead looking at an ensemble of classifiers that are trained using some quantum subroutine.  These classifiers may be trained with legitimate data or instead may be trained using data from an adversary.  We can envision a classifier, in the worst case, as being compromised if it receives the adversary's training set.  As such for our purposes we will mainly look at bagging through the lens of boosting, which uses an ensemble of different classifiers to assign a class to a test vector each of which may be trained using the same or different data sets.  Quantum boosting has already been discussed in~\cite{schuld2017quantum}, however, our approach to boosting differs considerably from this treatment. 

The type of threat that we wish to address with our quantum boosting protocol is given below.
\begin{threat}
Assume that the user is provided a programmable oracle, $C$, such that $C\ket{j} \ket{x} = \ket{j} C_j \ket{x}$ where each $C_j$ is a quantum classifier that acts on the test vector $\ket{x}$. The adversary is assumed to control a fraction of $0\le \alpha<1$ of all classifiers $C_j$ that the oracle uses to classify data and wishes to affect the classes reported by the user of the boosted quantum classifier that implements $C$.  The adversary is assumed to have no computational restrictions and complete knowledge of the classifiers / training data used in the boosting protocol. \label{threat:2}
\end{threat}

While the concept of a classifier has a clear meaning in classical computing, if we wish to apply the same ideas to quantum classifiers we run into certain challenges.  The first such challenge lies with the encoding of the training vectors.  In classical machine learning training vectors are typically bit vectors.  For example, a binary classifier can then be thought of as a map from the space of test vectors to $\{-1,1\}$ corresponding to the two classes.  This is if $C$ were a classifier then $Cv = \pm v$ for all vectors $v$ depending on the membership of the vector.  Thus every training vector can be viewed as an eigenvector with eigenvalue $\pm 1$.  This is illustrated for unit vectors in~\fig{eigenspace}.

Now imagine we instead that our test vector is a quantum state $\ket{\psi}$ in $\mathbb{C}^2$.  Unlike the classical setting, it may be physically impossible for the classifier to know precisely what $\ket{\psi}$ is because measurement inherently damages the state.  This makes the classification task fundamentally different than it is in classical settings where the test vector is known in its entirety.  Since such vectors live in a two-dimensional vector space and they comprise an infinite set, not all $\ket{\psi}$ can be eigenvectors of the classifier $C$.  However, if we let $C$ be a classifier that has eigenvalue $\pm 1$ we can always express  $\ket{\psi} = a\ket{\phi_+} + \sqrt{1-|a|^2} \ket{\phi_-}$, where $C\ket{\phi_\pm} =\pm \ket{\phi_\pm}$.  Thus we can still classify in the same manner, but now we have to measure the state repeatedly to determine whether $\ket{\psi}$ has more projection onto the positive or negative eigenspace of the classifier.  This notion of classification is a generalization of the classical idea and highlights the importance of thinking about the eigenspaces of a classifier within a quantum framework.

Our approach to boosting and bagging embraces this idea.  The idea is to combine an ensemble of classifiers to form a weighted sum of classifiers $C = \sum_j b_j C_j$ where $b_j>0$ and $\sum_{j} b_j =1$. Let $\ket{\phi}$ be a simultaneous $+1$ eigenvector of each $C_j$, ie $C_j \ket{\phi} = \ket{\phi}$
\begin{equation}
C \ket{\phi} =\sum_j b_j C_j \ket{\phi} = \ket{\phi}.
\end{equation}
The same obviously holds for any $\ket{\phi}$ that is a negative eigenvector of each $C_j$.
That is any vector that is in the simultaneous positive or negative eigenspace of each classifier will be deterministically classified by $C$.

\begin{figure}[t!]
\includegraphics[width=0.4\linewidth]{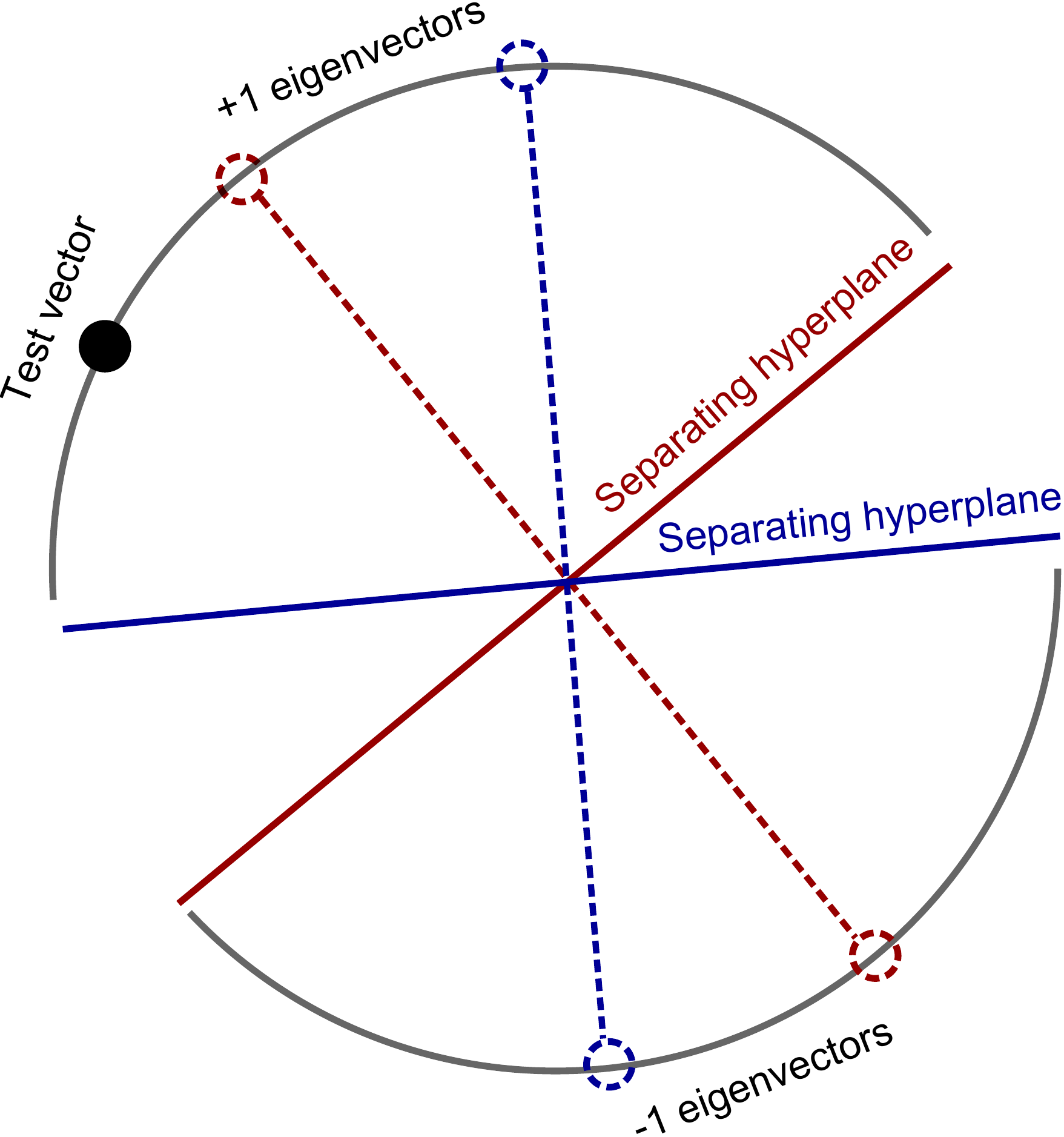}
\caption{Eigenspaces for two possible linear classifiers for linearly separated unit vectors.  Each trained model provides a different separating hyperplane and eigenvectors.  Bagging aggragates over the projections of the test vector onto these eigenvectors and outputs a class based on that aggragation.\label{fig:eigenspace} }
\end{figure}

The simplest way to construct a classifier out of an ensemble of classifiers is to project the test vector onto the eigenvectors of the sum of the classifiers and compute the projection of the state into the positive and negative eigenvalue subspaces of the classifier.  This gives us an additional freedom not observed in classical machine learning. While the positive and negative eigenspaces are fixed for each classifier in the classical analogues of our classifier, here they are not.  Also we wish our algorithm to function for states that are fed in a streaming fashion.  That is to say, we do not assume that when a state $\ket{\psi}$ is provided that we can prepare a second copy of this state.  This prevents us from straight forwardly measuring the expectation of each of the constituent $C_j$ to obtain a classification for $\ket{\psi}$.  It also prevents us from using quantum state tomography to learn $\ket{\psi}$ and then applying a classifier on it.  We provide a formal definition of this form of quantum boosting or bagging below.

\begin{definition}
Two-class quantum boosting or bagging is defined to be a process by which an unknown test vector is classified by projecting it onto the eigenspace of $C = \sum_j b_j C_j$ and assigning the class to be $+1$ if the probability of projecting onto the positive eigenspace is greater than $1/2$ and $-1$ otherwise.  Here each $C_j$ is unitary with eigenvalue $\pm 1$, $\sum_{j} b_j =1$, $b_j\ge 0$, and at least two $b_j$ that correspond to distinct $C_j$ are positive.\label{def:C}
\end{definition}

We realize such a classifier via phase estimation.  But before discussing this, we need to abstract the input to the problem for generality.  We do this by assuming that each classifier, $C_j$, being included is specified only by a weight vector $w_j$.  If the individual classifiers were neural networks then the weights could correspond to edge and bias weights for the different neural networks that we would find by training on different subsets of the data.  Alternatively, if we were forming a committee of neural networks of different structures then we simply envision taking the register used to store $w_j$ to be of the form $[{\rm tag},w_j]$ where the tag tells the system which classifier to use.  This allows the same data structure to be used for  arbitrary classifiers.

One might wonder why we choose to assign the class based on the mean class output by $C$ rather than the class associated with the mean of $C$.  In other words, we could have defined our quantum analogue of boosting such that the class is given by the expectation value of $C$ in the test state, ${\rm sign}(\bra{\psi} C\ket{\psi})$.  In such a case, we would have no guarantee that the quantum classifier would be protected against the adversary.  The reason for this is that the mean that is computed is not robust.  For example, assume that the ensemble consists of $N$ classifiers such that for each $C_j$, $|\bra{\psi}C_j\ket{\psi}|\le 1/(2N)$ and assume the weights are uniform (if they are not uniform then the adversary can choose to replace the most significant classifier and have an even greater effect).  Then if an adversary controls a single classifier example and knows the test example then they could replace $C_1$ such that $|\bra{\psi} C_1 \ket{\psi}|=1$.  The assigned class is then
\begin{equation}
{\rm sign}\left(\frac{\bra{\psi} C_1 \ket{\psi}}{N}-\frac{1}{2N}\right)< {\rm sign}\left(\frac{\bra{\psi} C_1 \ket{\psi}}{N}+\sum_{j=2}^{N}\frac{\bra{\psi} C_j \ket{\psi}}{N} \right)<{\rm sign}\left(\frac{\bra{\psi} C_1 \ket{\psi}}{N}+\frac{1}{2N} \right).
\end{equation}
In such an example, even a single compromised classifier can impact the class returned because by picking $\bra{\psi}C_1 \ket{\psi} = \pm 1$ the adversary has complete control over the class returned.  In contrast, we show that compromising a single quantum classifier does not have a substantial impact on the class if our proposal is used assuming the eigenvalue gap between positive and negative eigenvalues of $C$ is sufficiently large in the absence of an adversary.

We formalize these ideas in a quantum setting by introducing  quantum blackboxes, $\mathcal{T}$, $\mathcal{B}$ and $\mathcal{C}$.  The first takes an index and outputs the corresponding weight vector.  The second prepares the weights for the different classifiers in the ensemble $C$.  The final blackbox applies, programmed via the weight vector, the classifier on the data vector in question.  We formally define these black boxes below.

\begin{definition}
Let $\mathcal{C}$ be a unitary operator such that if $\ket{w}\in \mathbb{C}^{2^{n_w}}$ represents the weights that specify the classifier $\mathcal{C} \ket{w}\ket{\psi}= \ket{w}C(w)\ket{\psi}$  for any state $\ket{\psi}$. Let $\mathcal{T}$ be a unitary operator that, up to local isometries performs $\mathcal{T} \ket{j}\ket{0} = \ket{j}\ket{w_j}$, which is to say that it generates the weights for classifier $j$ (potentially via training on a subset of the data).  Finally,  define unitary $B: B\ket{0} = \sum_j \sqrt{b_j} \ket{j}$ for non-negative $b_j$ that sum to $1$.
\end{definition}

In order to apply phase estimation on $C$, for a fixed and unknown input vector, $\ket{\psi}$, we need to be able to simulate $e^{-iC t}$.  Fortunately, because each $C_j$ has eigenvalue $\pm 1$ it is easy to see that $C_j^2 =1$ and hence is Hermitian.  Thus $C$ is a Hamiltonian.  We can therefore apply Hamiltonian simulation ideas to implement this operator and formally demonstrate this in the following lemma.

\begin{lemma}
Let $C= \sum_{j=1}^M b_j C_j$ where each $C_j$ is Hermitian and unitary and has a corresponding weight vector $w_j$ for $b_j\ge 0~\forall~j$ and $\sum_j b_j =1$.  Then for every $\epsilon>0$ and $t\ge 0$ there exists a protocol for implementing a non-unitary operator, $W$, such that for all $\ket{\psi}$  $\|e^{-iC t}\ket{\psi}-W\ket{\psi}\|\le \epsilon$  using a number of queries to $\mathcal{C}$, $B$ and $\mathcal{T}$ that is in $O(t\log(t/\epsilon)/\log\log(t/\epsilon))$.\label{lem:sim}
\end{lemma}
\begin{proof}
The proof follows from the truncated Taylor series simulation algorithm. By definition we have that $B\ket{0}= \sum_j \sqrt{b_j} \ket{j}$.  Now let us consider the operation select$(V):= \mathcal{T}^\dagger \mathcal{C} \mathcal{T}$.  This operation has the effect that, up to isometries, $\mathcal{C}\ket{j} \ket{\psi} = \ket{j} C_j \ket{\psi}$.  The operations ${\rm select}(V)$ and $B$ take the exact same form as their analogues in the Taylor-series simulation algorithm.  If $T=\sum_j |b_j| t$ then the Taylor series simulation algorithm requires an expansion up to order $K\in O(\log(T/\epsilon)/\log\log(T/\epsilon))$ and the use of robust oblivious amplitude estimation requires $r$ such segments where $r\in O(T)$.  Thus the total number of queries to select$(V)$ made in the simulation is $O(Kr)= O(T\log(T/\epsilon)/\log\log(T/\epsilon)$, and the number of queries to $B$ in the algorithm is $O(r)\subseteq O(T)$.  

Next, by assumption we have that $\sum_j b_j=1$ and therefore $T=t$.  Thus the number of queries to select$(V)$ and $B$ in the algorithm is in $O(t\log(t/\epsilon)/\log\log(t/\epsilon))$.  The result then follows after noting that a single call to select$(V)$ requires $O(1)$ calls to $V$ and $C$.
\end{proof}

Note that in general purpose simulation algorithms, the application of select$(V)$ will require a complex series of controls to execute; whereas here the classifier is made programmable via the $\mathcal{T}$ oracle and so the procedure does not explicitly depend on the number of terms in $C$ (although in practice it will often depend on it through the cost of implementing $\mathcal{T}$ as a circuit).  For this reason the query complexity cited above can be deceptive if used as a surrogate for the time-complexity in a given computational model.

With the result of~\lem{sim} we can now proceed to finding the cost of performing two-class bagging or boosting where the eigenvalue of $C$ is used to perform the classification.  The main claim is given below.
\begin{theorem}
Under the assumptions of~\lem{sim}, the number of queries needed to $\mathcal{C}$, $B$ and $\mathcal{T}$ to project $\ket{\psi}$ onto an eigenvector of $C$ with probability at least $1-\delta$ and estimate the eigenvalue within error $\epsilon$ is in $O((\frac{1}{\delta\epsilon})\log(1/\epsilon)/\log\log(1/\epsilon))$.
\end{theorem}
\begin{proof}
Since $\|C\|\le 1$, it follows that we can apply phase estimation on the $e^{-iC}$ in order to project the unknown state $\ket{\psi}$ onto an eigenvector.  The number of applications of $e^{-iC}$ needed to achieve this within accuracy $\epsilon$ and error $\delta$ is $O(1/\delta\epsilon)$ using coherent phase estimation~\cite{WKS15}.  Therefore, after taking $t=1$ and using the result of~\lem{sim} to simulate $e^{-iC}$ that the overall query complexity for the simulation is in $O((\frac{1}{\delta\epsilon})\log(1/\epsilon)/\log\log(1/\epsilon))$.
\end{proof}

Finally, there is the question of how large on an impact an adversary can have on the eigenvalues and eigenvectors of the classifier.  This answer can be found using perturbation theory to estimate the derivatives of the eigenvalues and eigenvectors of the classifier $C$ as a function of the maximum fraction of the classifiers that the adversary has compromised.  This leads to the following result.
\begin{corollary}
Assume that $C$ is a classifier defined as per~\defn{C} that is described by a Hermitian matrix that the absolute value of each eigenvalue is bounded below by $\gamma/2$ and that one wishes to perform a classification of a data set based on the mean sign of the eigenvectors in the suport of a test vector.  Given that an adversary controls a fraction of the classifiers equal to $\alpha<\gamma/4$, the adversary cannot affect the classes output by this protocol in the limit where $\epsilon\rightarrow 0$ and $\delta \rightarrow 0$.
\end{corollary}
\begin{proof}
If an adversary controls a fraction of data equal to $\alpha$ then we can view the perturbed classifier as 
\begin{equation}
C' = \sum_{j=1}^{(1-\alpha) N} b_j C_j + \sum_{j=(1-\alpha)N+1}^N b_j C'_j= C + \sum_{j=1}^N \delta_{j> (1-\alpha)N} b_j (C'_j-C_j).
\end{equation}
This means that we can view $C'$ as a perturbation of the original matrix by a small amount.  In particular, this implies that because each $C_j$ is Hermitian and unitary
\begin{equation}
\|C' -C\| \le \alpha  \max_j\|C_j'-C_j\| \le 2\alpha .
\end{equation}
Using perturbation theory, (see~\eq{eigderiv0} in the appendix), we show that the maximum shift in any eigenvalue due to this perturbation is at most $2\alpha$. By which we mean that if $E_j$ is the $j^{\rm th}$ eigenvalue of $C$ and $C'(\sigma)= C + \sigma \sum_{j=1}^N \delta_{j> (1-\alpha)N} b_j (C'_j-C_j)$ then the corresponding eigenvalue $E_j(\sigma)$ obeys $|E_j(\sigma) -E_j| \le 2\sigma \alpha + O(\sigma^2)$.  By integrating this we find that $E_j(1)$ (informally the eigenvalue of $C'$ corresponding to the $j^{\rm th}$ eigenvalue of $C$) we find that $|E_j(1) - E_j|=|E_j'-E_j| \le 2\alpha$.  

Given that $E_j\ge \gamma/2$ the above argument shows that ${\rm sign}(E_j) = {\rm sign}(E_j')$ if $2\alpha <\gamma/2$.  Thus if the value of $\alpha$ is sufficiently small then none of the eigenvectors returned will have the incorrect sign.  This implies that if we neglect errors in the estimation of eigenvalues of $C$ then the adversary can only impact the class output by $C$ if they control a fraction greater than $\gamma/4$.  In the limit as $\delta\rightarrow 0$ and $\epsilon \rightarrow 0$ these errors become negligible and the result follows.
\end{proof}

From this we see that adversaries acting in accordance with \threatM{2} can be thwarted using boosting or bagging that have control over a small fraction of the quantum classifiers used in a boosting protocol.  This illustrates that by generalizing the concept of boosting to a quantum setting we can not only make our classifiers better but also make them more resilient to adversaries who control a small fraction of the classifiers provided.  This result, combined with our previous discussion of robust quantum PCA, shows that quantum techniques can be used to defend quantum classifiers against causative attacks by quantum adversaries.

 However, we have not discussed exploratory attacks which often could be used as the precursors to such attacks or in some cases may be the goal in and of itself.  We show below that quantum can be used in a strong way to defend against some classes of these attacks.

\section{Quantum Enhanced Privacy for Clustering}
Since privacy is one of the major applications of quantum technologies, it should come as no surprise that quantum computing can help boost privacy in machine learning as well.  As a toy example, we will first discuss how quantum computing can be used to allow $k$--means clustering to be performed without leaking substantial information about any of the training vectors.

The $k$--means clustering algorithm is perhaps the most ubiquitous algorithm used to cluster data.  While there are several variants of the algorithm, the most common variant attmepts to break up a data set into $k$ clusters such that the sum of the intra-cluster variances is minimized.  In particular, let $f:(x_j,\{\mu\})\mapsto (1,\ldots, k)$ gives the index of the set of centroids $\{\mu\}=\{\mu_1,\ldots,\mu_k\}$  and $x_j$ are the vectors in the data set.  The $k$--means clustering algorithm then seeks to minimize $\sum_j |x_j - \mu_{f(x_j,\{\mu\})}|_2^2$.  Formally, this problem is \NP--hard which means that no efficient clustering algorithm (quantum or classical) is likely to exist.  However, most clustering problems are not hard examples, which means that clustering generically is typically not prohibitively challenging.

The $k$--means algorithm for clustering is simple.  First begin by assigning the $\mu_p$ to data points sampled at random.  Then for each $x_j$ find the $\mu_p$ that is closest to it, and assign that vector to the $p^{\rm th}$ cluster.  Next set $\mu_1,\ldots,\mu_k$ to be the cluster centroids of each of the $k$ clusters and repeat the previous steps until the cluster centroids converge.

A challenge in applying this algorithm to cases, such as clustering medical data, is that the individual performing the algorithm typically needs to have access to the users information.  For sensitive data such as this, it is difficult to apply such methods to understand structure in medical data sets that are not already anonymized.  Quantum mechanics can be used to address this.  

Imagine a scenario where an experimenter wishes to collaboratively cluster a private data set.  The experimenter is comfortable broadcasting information about the model to the $N$ owners of the data, but the users are not comfortable leaking more than $m$ bits of information about their private data in the process of clustering the data.  Our approach, which is related to quantum voting strategies~\cite{hillery2006towards}, is to share an entangled quantum state between the recipients and use this to anonymously learn the means of the data.

\begin{threat}
Assume that a group of $N$ users wish to apply $k$--means clustering to cluster their data and that an adversary is present that wishes to learn the private data held by at least one of the users as an exploratory attack and has no prior knowledge about any of the users' private data before starting the protocol.  The adversary is assumed to have control over any and all parties that partake in the protocol, authentication is impossible in this setting and the adversary is unbounded computationally.
\label{threat:3}
\end{threat}

The protocol that we propose to thwart such attacks is given below.

\begin{enumerate}
\item The experimenter broadcasts $k$ cluster centroids over a public classical channel.
\item For each cluster centroid, $\mu_p$, the experimenter sends the participant one qubit out of the state  $(\ket{0}^{N} + \ket{1}^N)/\sqrt{2}$.
\item Each participant that decides to contribute applies $e^{- iZ/2N}$ to the qubit corresponding to the cluster that is closest to their vector.  If two clusters are equidistant, the closest cluster is chosen randomly.
\item The experimenter repeats above two steps $O(1/\epsilon)$ times in a phase estimation protocol to learn $P(f(x_j,\{\mu\})=p)$ within error $\epsilon$.  Note that because participants may not participate $\sum_{p=1}^k P(f(x_j,\{\mu\})=p)\le 1$.
\item Next the experimenter performs the same procedure except phase estimation is now performed for each of the $d$ components of $x_j$, a total of $O(d/\epsilon P(f(x_j)=p))$ for each cluster $p$.
\item From these values the experimenter updates the centroids and repeats the above steps until convergence or until the privacy of the users data can no longer be guaranteed.
\end{enumerate}

Intuitively, the idea behind the above method is that each time a participant interacts with their qubits they do so by performing a tiny rotation.  These rotations are so small that individually they cannot be easily distinguished from performing no rotation, even when the best test permitted by quantum mechanics is used.  However, collectively the rotation angles added by each participant sums to non-negligible rotation angle if a large number of participants are pooled.  By encoding their private data as rotation angles, the experimenter can apply this trick to learn cluster means without learning more than a small fraction of a bit of information about any individual participant's private data.

\begin{lemma}
Steps 1-5 of the above protocol take a set of $d$--dimensional vectors $\{x_j:j=1,\ldots, N\}$ held by $N$ potential participants and a set of cluster centroids $\{\mu_p: p=1,\ldots,k\}$ and computes an iteration of $k$--means clustering with the output cluster centroids computed within error $\epsilon$ in the infinity norm and each participant requires a number of single qubit rotations that scales as $O(d/[\min_p P(f(x_j,\{\mu\})=p) \epsilon])$ under the assumption that each participant follows the protocol precisely.\label{lem:rotbound}
\end{lemma}
\begin{proof}
Let us examine the protocol from the perspective of the experimenter.  From this perspective the participants' actions can be viewed collectively as enacting blackbox transformations that apply an appropriate phase on the state $(\ket{0}^N +\ket{1}^N)/\sqrt{2}$.  Let us consider the first phase of the algorithm, corresponding to steps $2$ and $3$, where the experimenter attempts to learn the probabilities of users being in each of the $k$ clusters.

First when the cluster centroids are announced to participant $j$, they can then classically compute the distance $x_j$ and each of the cluster centroids efficiently.  No quantum operations are needed to perform this step.  Next for the qubit corresponding to cluster $p$ user $j$ performs a single qubit rotation and uses a swap operation to send that qubit to the experimenter.  This requires $O(1)$ quantum operations.  The collective phase incurred on the state $(\ket{0}^N +\ket{1}^N)/\sqrt{2}$ from each such rotation results in (up to a global phase)
$$
(\ket{0}^N +\ket{1}^N)/\sqrt{2}\mapsto (\ket{0}^N +e^{-i \sum_j \delta_{f(x_j,\{\mu\}),p}/N}\ket{1}^N)/\sqrt{2}=(\ket{0}^N +e^{-i P(f(x_j)=p)}\ket{1}^N)/\sqrt{2}.
$$
It is then easy to see that after querying this black box $t$-times that the state obtained by them applying their rotation $t$ times is $(\ket{0}^N +e^{-i P(f(x_j)=p) t}\ket{1}^N)/\sqrt{2}$.
Upon receiving this, the experimenter performs a series of $N$ controlled--not operations to reduce this state to $(\ket{0} +e^{-i P(f(x_j)=p)}\ket{1})/\sqrt{2}$ up to local isometries.  Then after applying a Hadamard transform the state can be expressed as $\cos(P(f(x_j)=p) t/2) \ket{0} + \sin(P(f(x_j)=p) t/2) \ket{1}$.  Thus the probability of measuring this qubit to be $0$ exactly matches that of phase estimation, wherein $P(f(x)=p)$ corresponds to the unknown phase.  This inference problem is well understood and solutions exist such that if $\{t_q\}$ are the values used in the steps of the inference process then $\sum_{q} |t_q|\in O(1/\epsilon)$  if we wish to estimate $P(f(x_j)=p)$ within error $\epsilon$.    

While it may na\i vely seem that the users require $O(k/\epsilon)$ rotations to perform this protocol, each user in fact only needs to perform $O(1/\epsilon)$ rotations.  This is because each vector is assigned to only one cluster using the above protocol.  Thus only $O(1/\epsilon)$ rotations are needed per participant.

In the next phase of the protocol, the experimentalist follows the same strategy to learn the means of the clustered data from each of the participants.  In this case, however, the blackbox function is of the form
$$
(\ket{0}^N +\ket{1}^N)/\sqrt{2}\mapsto (\ket{0}^N +e^{-i \sum_j [x_j]_q \delta_{f(x_j),p}/N}\ket{1}^N)/\sqrt{2},
$$
where $[x_j]_q$ is the $q^{\rm th}$ component of the vector.  By querying this black box $t$ times and performing the exact same transformations used above, we can apply phase estimation to learn $\sum_j [x_j]_q \delta_{f(x_j),p}/N$ within error $\delta$ using $O(1/\delta)$ queries to the participants.  Similarly, we can estimate 
$$\frac{\sum_j [x_j]_q \delta_{f(x_j),p}}{NP(f(x_j)=p)}= \frac{\sum_j [x_j]_q \delta_{f(x_j),p}}{\sum_j \delta_{f(x,j),p}}$$
within error $\epsilon$ using $O(1/[\epsilon P(f(x_j)=p)])$ querie  if $P(f(x_j)=p)$ is known with no error.

If $P(f(x_j)=p)$ is known within error $\epsilon$ and $P(f(x_j)=p)> \epsilon$ then it is straight forward to see that
$$
\left|\frac{1}{P(f(x_j)=p)}-\frac{1}{P(f(x_j)=p)+\epsilon} \right|\in O(\epsilon).
$$
Then because $\frac{\sum_j [x_j]_q \delta_{f(x_j),p}}{N}\in O(1)$ it follows that the overall error is at most $\epsilon$ given our assumption that $\min_p P(f(x_j)=p) > \epsilon$.  Each component of the centroid can therefore be learned using $O(1/[\min_p P(f(x_j)=p) \epsilon])$ rotations.
Thus the total number of rotations that each user needs to contribute to provide information about their cluster centroid is $O(d/[\min_{p} P(f(x_j)=p)\epsilon])$.

\end{proof}
While this lemma shows that the protocol is capable of updating the cluster centroids in a $k$--means clustering protocol, it does not show that the users' data is kept secret.  In particular, we need to show that even if the experimenter is an all powerful adversary then they cannot even learn whether or not a user actually participated in the protocol.  The proof of this is actually quite simple and is given in the following theorem.
\begin{theorem}
Let $B$ be a participant in a clustering task that repeats the above protocol $R$ times where $B$ has $x_j$ and let $\{\mu^{(r)}\}$ be the cluster centroids at round $r$ of the clustering.  Assume that the eavesdropper $E$ assigns a prior probability that $B$ participated in the above protocol to be $1/2$ and wishes to determine whether $B$ contributed $x_j$ (which may not be known to $E$) to the clustering.  The maximum probability, over all quantum strategies, of $E$ successfully deciding this is $1/2 +O(Rd/[\min_{p,r} P(f(x_j,\{\mu^{(r)}\})=p)N\epsilon])$ if $\min_{p,r} P(f(x_j,\{\mu^{(r)})=p) > \epsilon$. 
\end{theorem}
\begin{proof}
The proof of the theorem proceeds as follows.  Since $B$ does not communicate classically to the experimenter, the only way that $E$ can solve this problem is by feeding $B$ an optimal state to learn $B$'s data.  Specifically, imagine $E$ usurps the protocol and provides $B$ a state $\rho$.  This quantum state operator is chosen to maximize the information that $E$ can learn about the state.  The state $\rho$, for example, could be entangled over the multiple qubits that would be sent over the protocol and could also be mixed in general.  When this state is passed to $B$ a transformation $U\rho U^\dagger$ is enacted in the protocol.  Formally, the task that $E$ is faced with is then to distinguish between $\rho$ and $U\rho U^\dagger$.

Fortunately, this state discrimination problem is well studied.  The optimal probability of correctly distinguishing the two, given a prior probability of $1/2$, is~\cite{fuchs1999cryptographic}
\begin{equation}
P_{\rm opt} = \frac{1}{2}+ \frac{1}{4}{\rm Tr}|\rho - U\rho U^\dagger|.
\end{equation}
Let us assume without loss of generality that $B$'s $x_j$ is located in the first cluster.  Then $U$ takes the form
\begin{equation}
U=(e^{-i Z/2N})^{\otimes q_1} \otimes (e^{-i Z/2N})^{\otimes q_2} \otimes 1 = e^{-i [\sum_{p=1}^{q_1} Z^{(p)}+\sum_{r=q_1+1}^{q_1+q_2} Z^{(r)}]/2N}:=e^{-iH/2N},
\end{equation}
where $1$ refers to the identity acting on the qubits used to learn data about clusters $2$ through $k$ (as well as any ancillary qubits that $E$ chooses to use to process the bits) and $q_1$ and $q_2$ are the number of rotations used in the first and second phases discussed in~\lem{rotbound}.  $Z^{(p)}$ refers to the $Z$ gate applied to the $p^{\rm th}$ qubit.

Using Hadamard's lemma and the fact that ${\rm Tr}(|\rho A|) \le \|A\|_2$ for any density operator $\rho$ we then see that provided $(q_1+q_2)< N$
\begin{align}
{\rm Tr} \left|\rho - e^{-iH/2N} \rho e^{iH/2N}\right| &= {\rm Tr}\left|\frac{-i}{2N}[H,\rho] -\frac{1}{(2N)^22!} [H,[H,\rho]] +\cdots\right| \nonumber\\
&\le \sum_{n=1}^\infty \frac{\|H\|_2^n}{N^nn!}\in O\left(\frac{q_1+q_2}{N}\right).\label{eq:distbd}
\end{align}

Let $\{\mu^{(r)}\}$ be the cluster cetnroids in round $r$ of the $R$ rounds.  From~\lem{rotbound} we see that the number of rotations used in the $R$ rounds of the first phase obeys $q_1 \in O(R/\epsilon)$ and the number of rotations in the $R$ rounds of the second phase obeys $q_2 \in O(Rd/[\min_{p,r} P(f(x_j,\{\mu^{(r)}\})=p)\epsilon])$.  We therefore have from~\eq{distbd} that if $\min_{p,r} P(f(x_j,\{\mu^{(r)}\})=p) > \epsilon$ then
\begin{equation}
|P_{\rm opt} - 1/2| \in O\left(\frac{Rd}{\min_{p,r} P(f(x_j,\{\mu^{(r)}\})=p) N\epsilon}\right).
\end{equation}
\end{proof}

If the minimum probability of membership over all $k$ clusters and $R$ rounds obeys
\begin{equation}
\min_{p,r} P(f(x_j,\{\mu^{(r)}\})=p)\in\Omega(1/k),
\end{equation} which is what we expect in typical cases, the probability of a malevolent experimenter discerning whether participant $j$ contributed data to the clustering algorithm is $O(dk/N\epsilon)$.  It then follows that if a total probability of $1/2+\delta$ for the eavesdropper identifying whether the user partook in the algorithm can be tolerated then $R$ rounds of clustering can be carried out if $N \in \Omega\left( \frac{Rdk}{\epsilon\delta}\right)$.  This shows that if $N$ is sufficiently large then this protocol for $k$--means clustering can be carried out without compromising the privacy of any of the participants.

An important drawback of the above approach is that the scheme does not protect the model learned from the users.  Such protocols could be enabled by only having the experimenter reveal hypothetical cluster centroids and then from the results infer the most likely cluster centroids, however this diverges from the $k$-means approach to clustering and will likely need a larger value of $N$ to guarantee privacy given that the information from each round is unlikely to be as useful as it is in $k$--means clustering.

Also, while the scheme is private it is not secure.  This can easily be seen from the fact that an eavesdropper could intercept a qubit and apply a random phase to it.  Because the protocol assumes that the sum of the phases from each participant adds up to at most $1$, this can ruin the information sent.  While this approach is not secure against an all-powerful adversary, the impact that individual malfeasant participants could have in the protocol can be mitigated.  One natural way is to divide the $N$ participants into a constant number of smaller groups, and compute the median of the cluster centroids returned.  While such strategies will be successful at thwarting a constant number of such attacks, finding more general secure and private quantum methods for clustering data remains an important goal for future work.

\section{Conclusion}
We have surveyed here the impacts that quantum computing can have on security and privacy of quantum machine learning algorithms.  We have shown that robust versions of quantum principal component analysis exist that retain the exponential speedups observed for its non-robust analogue (modulo certain assumptions about data access and output).  We also show how bootstrap aggregation or boosting can be performed on a quantum computer and show that it is quadratically easier to generate statistics over a large number of weak classifiers using these methods.  Finally, we show that quantum methods can be used to perform a private form of $k$--means clustering wherein no eavesdropper can determine with high-probability that any participant contributed data, let alone learn that data.

These results show that quantum technologies hold promise for helping secure machine learning and artificial intelligences.  Going forward, providing a better understanding of the impacts that technologies such as blind quantum computing~\cite{arrighi2006blind,broadbent2009universal} may have for both securing quantum machine learning algorithms as well as blinding the underlying data sets used from any adversary.  Another important question is whether tools from quantum information theory could be used to bound the maximum information about the models used by quantum classifiers, such as quantum Boltzmann machines or quantum PCA, that adversaries can learn by querying a quantum classifier that is made publicly available in the cloud.  Given the important contributions that quantum information theory has made for our understanding of privacy and security in a quantum world, we have every confidence that these same lessons will one day have an equal impact on the security and privacy of machine learning.

There are of course many open questions left in this field and we have only attempted to give a taste here of the sorts of questions that can be raised when one looks at quantum machine learning in adversarial settings.  One important question surrounds the problem of generating adversarial examples for quantum classifiers.  Due to the unitarity of quantum computing, many classifiers can be easily inverted from the output quantum states.  This allows adversarial examples to be easily created that would generate false positives (or negatives) when exposed to the classifier.  The question of whether quantum computing offers significant countermeasures against such attacks remains an open problem.  Similarly, gaining a deeper understanding of the limitations that sample lower bounds on state/process imply on stealing quantum models imply for the security of quantum machine learning solutions could be an interesting question for further reasearch.  Understanding how quantum computing can provide ways of making machine learning solutions more robust to adversarial noise of course brings more than security, it helps us understand how to train quantum computers to understand concepts in a robust fashion, similar to how we understand concepts.  Thus thinking about such issues may help us address what is perhaps the biggest open question in quantum machine learning: ``Can we train a quantum computer to think like we think?''

\appendix
\section{Proof of Theorem 1 and Robust LCU Simulation}

Our goal in this approach to robust PCA is to examine the eigendecomposition of a vector in terms of the principal components of $\mathbf{M}$.  This has a major advantage over standard PCA in that it is manifestly robust to outliers.  We will see that, within this oracle query model, quantum approaches can provide great advantages for robust PCA if we assume that $\mathbf{M}$ is sparse. 

Now that we have an algorithm for coherently computing the median, we can apply this subroutine to compute the components of $\mathbf{M}$.  First we need a method to be able to compute a representation of $e_k^T x_j$ for any $k$ and $j$.  

\begin{lemma}
For any $k$ and $j$ define $\ket{\Psi_{j,k}}$ such that, for integer $y$, if $\braket{y}{\Psi_{j,k}}\ne 0$ then $| y- e_k^T v_j|\le \epsilon$ there exists a coherent quantum algorithm that maps, up to isometries, $\ket{j,k}\ket{0} \mapsto \ket{j,k} (\sqrt{1-\chi_{j,k}} \ket{\Psi_{j,k}} + \sqrt{\chi_{j,k}} \ket{\Psi_{j,k}^\perp})$ for $0\le \chi_{j,k} \le \Delta\le 1$ that uses $O(\log(1/\Delta)/\epsilon)$ queries to controlled $U$.\label{lem:ip}
\end{lemma}
\begin{proof}
First let us assume that all the $v_j$ are not unit vectors.  Under such circumstances we can use~\lem{isometry} to embed these vectors as unit vectors in a high dimensional space.  Therefore we can assume without loss of generality that $v_j$ are all unit vectors.

Since $v_j$ can be taken to be a unit vector, we can use quantum approaches to compute the inner product.  Specifically, the following circuit can estimate the the inner product $\braket{v_j}{k}$.
\[
    \Qcircuit @C=1em @R=1em {
        \lstick{\ket{0}}    & \gate{H}     	& \ctrl{1}   			&\ctrl{2}& \gate{H} & \qw \\
        \lstick{\ket{j}} & {/} \qw                    		 & \multigate{1}{U} 	&\qw& \qw      & \qw   \\
        \lstick{\ket{0}} & {/} \qw                   		 & \ghost{U} 		&\targ& \qw      & \qw    \\
        \lstick{\ket{k}} & {/} \qw                      		& \qw 			&\ctrl{-1}& \qw      & \qw    
    }
\]
The circuit implements the Hadamard test on the unitary $U$, which prepares the state $\ket{v_j}$ and computes the bitwise exclusive or on the resultant vector and the basis vector $k$. 
Specifically, the probability of measuring the top-most qubit to be $0$ is
\begin{equation}
P(0|j,k) = \frac{1+ \braket{k}{v_j}}{2},
\end{equation}
where recall that we have assumed that the training vectors satisfy $\ket{v_j} \in \mathbb{R}$.  To see this, note that the circuit performs the following transformation
\begin{align}
\ket{0}\ket{j}\ket{0}\ket{k} &\mapsto (H\otimes \openone)\frac{1}{\sqrt{2}} \left(\ket{0}\ket{j}\ket{0}\ket{k} + \ket{1}\ket{j}\ket{v_j\oplus k} \ket{k} \right) \nonumber\\
&=\frac{1}{\sqrt{2}} \left(\frac{1}{\sqrt{2}} \ket{0}\ket{j}\left(\ket{0} + \ket{v_j\oplus k} \right)\ket{k} \right)+ \frac{1}{\sqrt{2}} \ket{1}\ket{\phi^\perp}:=\ket{\psi},
\end{align}
for some state vector $\ket{\phi^\perp}$.  The probability of measuring the first qubit to be $0$ is ${\rm Tr}( \ketbra{\psi}{\psi} (\ketbra{0}{0}\otimes \openone))$ which after some elementary simplifications is 
\begin{equation}
\left(\frac{\bra{0}+\bra{v_j\oplus k}}{{2}} \right)\left(\frac{\ket{0}+\ket{v_j\oplus k}}{{2}} \right) = \frac{1+\braket{0}{v_j\oplus k}}{2}=\frac{1+\braket{k}{v_j}}{2},
\end{equation}
given that $v_j$ is real valued.

Following the argument for coherent amplitude estimation in~\cite{WKS15} we can use this process to produce a circuit that maps $\ket{j,k}\ket{0} \mapsto \ket{j,k} (\sqrt{1-\chi_{j,k}} \ket{\tilde{\Psi}_{j,k}} + \sqrt{\chi_{j,k}} \ket{\tilde{\Psi}_{j,k}^\perp}$ for $0\le \chi_{j,k} \le \Delta/2\le 1$ using $O(\log(1/\Delta)/\epsilon)$ invocations of the above circuit where, for integer $y$, $\braket{y}{\tilde{\Psi}}\ne 0$ if $|y- (1+e_k^T v_j)/2|\le \epsilon/2$.  Each query requires $1$ invocation of controlled $U$ and hence the query complexity is $O(\log(1/\Delta)/\epsilon)$.  By subtracting off $1/2$ from the resulting bit string and multiplying the result by $2$ we obtain the desired answer.  The result is computing to precision $\epsilon$ because we demanded that $(1+e_k^T v_j)/2$ is computed with accuracy $\epsilon/2$.  The arithmetic requires no additional queries, so the overall query complexity is $O(\log(1/\Delta)/\epsilon)$ as claimed.  Note that in this case there are several junk qubit registers that are invoked by this algorithm that necessarily extend the space beyond that claimed in the lemma statement; however, since we only require the result to be true up to isometries we neglect such registers above for simplicity.
\end{proof}

\begin{theorem}
Let $U$ be a unitary operation such $U\ket{0}\ket{j}\ket{0} = \ket{0}\ket{j}\ket{0}$ and $U\ket{1}\ket{j}\ket{0} = \ket{1}\ket{j}\ket{x_j}$.  
Let $Q:[0,1] \mapsto [-1,1]$ be an inverse cummulative distribution function for the $x_j$ and assume that $Q$ is Lipschitz with constant  $L>0$.
For any $\delta>0$ a unitary $\mathcal{V}$ can be implemented such that 
$
\mathcal{V}\ket{j}\ket{k}\ket{0} = \ket{j}\ket{k}\left(\sqrt{1-\delta}\ket{\mathcal{P}} + \sqrt{\delta} \ket{\mathcal{P}^\perp}\right).
$
Here $\ket{\mathcal{P}}$ be a quantum state on a Hilbert space $\mathcal{H}_{\rm junk} \otimes \mathcal{H}_{N_v}$ such that
${\rm Tr} (\ketbra{y}{y}{\rm Tr_{junk}} (\ketbra{\mathcal{P}}{\mathcal{P}})) \ne 0$ if and only if $|P(\braket{x_k}{x_j} <y)- 1/2|\le \epsilon<1/4$.  Furthermore, this state can be prepared using a number of queries to $U$ that scales as $$O(\epsilon^{-2} \log(1/\epsilon)\log(L/\epsilon) \log(\log(1/\epsilon)/\delta)).$$\label{thm:PCAmain}
\end{theorem}
\begin{proof}
Similar to the proof of Grover's algorithm for the median~\cite{grover1996fast} and that of Nayak and Wu~\cite{nayak1999quantum}, our approach reduces to coherently applying binary search.  Their algorithms cannot be directly used here because they utilize measurement, which prevents their use within a coherent amplitude estimation algorithm.  Furthermore, the algorithms of~\cite{nayak1999quantum} solve a harder problem, namely that of outputting an approximate median from the list rather than simply providing a value that approximately seperates the data into two equal length halves.  This value, however, need not actually be contained in either list unlike the algorithm of~\cite{nayak1999quantum}.  For this reason we propose a slight variation on Nayak and Wu's ideas here.

Consider the following algorithm for some value of $\epsilon'>0$
\begin{enumerate}
\item Prepare the state $\ket{L_1}\ket{R_1}$ such that $L_1$ corresponds to $0$ and $R_1$ corresponds to $1$.
\item Repeat steps 3 through 8 for $p=1$ to $p=p_{\max}\in O(\log(1/\epsilon))$
\item $\ket{L_p}\ket{R_p}\rightarrow \frac{1}{\sqrt{N}}\sum_{j=1}^N \ket{j}\ket{L_p}\ket{R_p}\ket{\mu_p}\ket{0}\ket{0}$ for $\mu_p = ([L_p +R_p]/2)$.
\item Repeat the following step within a coherent amplitude estimation with error $\epsilon_0/2$ and error probability $\delta_0$ on the fifth register using a projector that marks all states in that register, $\ket{L_p}$, such that $L_p< \mu_p$ and store the probabilities in the sixth register.
\item Apply coherent amplitude estimation on the inner product circuit in Lemma $2$ to prepare state
$$ \frac{1}{\sqrt{N}}  \sum_{j=1}^N \ket{j}\ket{L_p}\ket{R_p} \ket{\mu_p}\left(\sqrt{1-\delta_p(j,k)} \ket{\psi_{p}(j,k)} +\sqrt{\delta_p(j,k)}\ket{\psi_p^\perp(j,k)} \right)\ket{0}, $$
where $ \braket{\nu}{\psi_{p}(j,k)}\ne 0$ if $|\nu-|\braket{x_k}{x_j}|^2| \le \epsilon'$ and $\delta_{p}(j,k) \le \epsilon_0/2$.
\item Use a reversible circuit to set, conditioned on the probability computed in the above steps is less than $1/2$, $\ket{L_{p+1}} \gets \ket{\mu_p}$ and $\ket{R_{p+1}} \gets \ket{R_p}$.
\item Use a reversible circuit to set, conditioned on the probability computed in the above steps is greater than $1/2$, $\ket{R_{p+1}} \gets \ket{\mu_p}$ and $\ket{L_{p+1}} \gets \ket{L_p}$.
\item Use a reversible circuit to set, conditioned on whether $L$ or $R$ was updated $\ket{L_{p+1}} \mapsto \ket{L_{p+1} - (\epsilon + L \epsilon_0)}$ or  $\ket{R_{p+1}} \mapsto \ket{R_{p+1} + (\epsilon + L \epsilon_0)}$
\item Return final state wherein the register $\ket{\mu_{p_{\max}+1}}$ contains the estimate of the median.

\end{enumerate}
Let us focus on a single update step in the outer loop, which constitutes estimation of the inner products and computation of a median.  \lem{ip} shows that there exists a unitary process, that performs (up to isometries) 
\begin{equation}
O_{ip}:\ket{j,k}\ket{0} \mapsto \ket{j,k} \left(\sqrt{1-\chi_{j,k}} \ket{\psi_{j,k}} + \sqrt{\chi_{j,k}} \ket{\psi_{j,k}^\perp} \right).
\end{equation}
Now define an operator $O_\mu$ such that for integers $\mu$ and $y$
\begin{equation}
O_\mu \ket{y}\ket{\mu} = \begin{cases}\ket{y}\ket{\mu} &, \text{ if } \mu>y \\ -\ket{y}\ket{\mu} &, \text{ otherwise }  \end{cases}.
\end{equation}
This oracle serves to mark all states that have support on values less than or equal to $y$.  We can then express the state $O_{ip}\ket{j,k}\ket{0}$ as
\begin{equation}
\ket{j,k} \left(\sqrt{1-\chi_{j,k}} \left(\sqrt{P_{<}'}\ket{\Psi_{j,k}^<} + \sqrt{1-P_<'} \ket{{\Psi_{j,k}^>}} \right)+\sqrt{L_{j,k}}\ket{\phi_{j,k}^\perp} \right):=\ket{j,k}\ket{\Upsilon_{j,k}},
\end{equation}
where $O_{\mu} \ket{\Psi_{j,k}^<}\ket{\mu} = -\ket{\Psi_{j,k}^<}\ket{\mu}$, $O_\mu \ket{\Psi_{j,k}^>}\ket{\mu} = \ket{\Psi_{j,k}^>}\ket{\mu}$ and $ \ket{\phi_{j,k}^\perp}$ is not necessarily an eigenvector of $O_{\mu}$ and we formally express $O_{\mu}\ket{\phi_{j,k}^\perp}\ket{\mu} = \ket{\tilde \phi_{j,k}^\perp}\ket{\mu}$.  If we apply coherent amplitude estimation using the oracles $O_{\mu}$ and $O_{ip}$ with error at most $\epsilon_0/2$ and failure probability at most $\delta_0\ge \Delta_{j,k}$ we obtain a state of the form
\begin{equation}
\ket{j,k} \left(\sqrt{1-\Delta_{j,k}} \ket{P_{\rm good}^{j,k}} + \sqrt{\Delta_{j,k}} \ket{P_{\rm bad}^{j,k}}\right)\ket{\mu},
\end{equation}
where $\braket{P_{\rm good}^{j,k}}{x}\ne 0$ only if $$\bra{\Upsilon_{j,k}}\bra{\mu} [(\openone -O_\mu)/2]\ket{\Upsilon_{j,k}}\ket{\mu} =  |(1-\chi_{j,k})P_{<}' + \chi_{j,k} \|(\openone - O_{\mu})/2 \ket{\tilde \phi_{j,k}^\perp}\| -x| \le \epsilon_0/2.$$  This is equivalent to saying that the estimate of the probability output by the process, in the success branch of the wave function, matches the exact probability within error at most $\epsilon_0/2$.

Now we further have from the fact that $\chi_{j,k} \le \epsilon_0/2$ that
\begin{align}
|P_<' - x| &\le |P_<' - (1-\chi_{j,k})P_{<}' + \chi_{j,k} \|(\openone - O_{\mu})/2 \ket{\tilde \phi_{j,k}^\perp}\|| + |(1-\chi_{j,k})P_{<}' + \chi_{j,k} \|(\openone - O_{\mu})/2 \ket{\tilde \phi_{j,k}^\perp}\| -x|\nonumber\\
&\le |\chi_{j,k}P_{<}' - L_{j,k} \|(\openone - O_{\mu})/2 \ket{\tilde \phi_{j,k}^\perp}\|| +\epsilon_0/2 \le \epsilon_0.
\end{align}
Thus the errors in the probability estimates that the median, or more generally any percentile, is based on are at most $\epsilon_0$ with the assumed level of precision.

Now without loss of generality, let us assume that the right end point is updated in a given step of the protocol.  Let $R$ be the value of the right-endpoint in the event that there was zero error in the estimate of the probability and let $\tilde{R}$ be the approximation to the right end point that arises from errors in the median probability and also from estimation of the inner products.  If we then define $\tilde{R}_p$ to be the error in the estimate of $R$ that arises when we only consider errors from the estimate of the median probability we see that
\begin{equation}
|R-\tilde{R}| \le |R-\tilde{R}_p| +|\tilde{R}_p - \tilde{R}|.
\end{equation}
The maximum deviation that can occur in the $p^{\rm th}$ percentile from perturbing each element of a list by $\epsilon'$ is $\epsilon'$.  Similarly, since $Q$ is Lipshitz with constant $L$ we have that $|R-\tilde{R}_p| \le L \epsilon_0$.  Thus
\begin{equation}
|R-\tilde{R}| \le \epsilon'+L \epsilon_0.
\end{equation}
Note that in many cases, the error in the right end point will be zero because the uncertainty in the probability estimation will be less than the gap between the percentiles at $R$ and $L$; however, we will track this propagation of uncertainty in each step for simplicity.
Thus we have that if $\epsilon_t(p)$ is the total uncertainty in the median at step $p$ of the binary search procoess we have that, in the worst case scenario that
\begin{equation}
\epsilon_t(p+1) \le \frac{\epsilon_t(p)+L \epsilon_0+\epsilon'}{2}.
\end{equation}
Solving this recurrence relation with $\epsilon_t(0)=1/2$ yields,
\begin{equation}
\epsilon_t(n) \le 2^{-n-1} + (\epsilon'+L\epsilon_0)(1-2^{-n}).
\end{equation}
Thus if we take $\epsilon_0 = \epsilon'/L$ and we find that if we want the total uncertainty to be $\epsilon$ at the end of the protocol then it suffices to pick

$$n=\log_2 \left(\frac{1-4\epsilon}{2(\epsilon-4\epsilon')} \right).$$
If $\epsilon>4\epsilon'$ and $\epsilon < 1/4$ then it is clear that it suffices to take $p_{\max}\in O(\log(1/\epsilon))$ in order to obtain an uncertainty in the mean of $\epsilon$ at the end of the binary search.  We assume in the theorem statement that $\epsilon<1/4$ and choose $\epsilon' < \epsilon/4$.  This validates the claim that the binary search process need only be repeated $O(\log(1/\epsilon))$ times.

We then have that the overall query complexity of the algorithm is the product of the query complexities of performing the inner product calculation, calculating the median and the number of repetitions needed in binary search.  The final complexity, $C$, therefore scales from~\lem{ip} and from bounds on the query complexity of coherent amplitude estimation~\cite{WKS15} as
\begin{equation}
C\in O\left(\frac{\log(1/\epsilon)\log(1/\epsilon_0)\log(1/\delta_0)}{\epsilon^2} \right)\in O\left(\frac{\log(1/\epsilon)\log(L/\epsilon)\log(\log(1/\epsilon)/\delta)}{\epsilon^2} \right).
\end{equation}
Here we have used the fact that the success probability for the overall algorithm is at least $1-\delta$ therefore we need to pick $\delta_0 \in O(\delta/\log(1/\epsilon))$ since there are $O(\log(1/\epsilon))$ distinct rounds and the failure probability grows at most additively from the union bound.

Finally, note that in none of the steps in the algorithm has any garbage collection of ancillary qubits been performed.  For this reason it is necessary to explicitly write the resultant state as
$$
\ket{j}\ket{k}\left(\sqrt{1-\delta}\ket{\mathcal{P}} + \sqrt{\delta} \ket{\mathcal{P}^\perp}\right),
$$
where $\ket{\mathcal{P}}$ be a quantum state on a Hilbert space $\mathcal{H}_{\rm junk} \otimes \mathcal{H}_{N_v}$ such that
${\rm Tr} (\ketbra{y}{y}{\rm Tr_{junk}} (\ketbra{\mathcal{P}}{\mathcal{P}})) \ne 0$ if and only if $|P(\braket{x_k}{x_j} <y)- 1/2|\le \epsilon$.
\end{proof}

\begin{corollary}
Given the assumptions of ~\thm{PCAmain} hold with $L\in \Theta(1)$, an oracle that provides $\gamma$ approximate queries to $\mathbf{M}_{k,\ell}$ can be implemented in the form of a unitary that maps for $1\ge\delta\ge \delta_{k,\ell}\ge 0$
$$
V_{\mathbf{M}}:\ket{k,\ell}\ket{0} \mapsto 
\ket{k,\ell}\sqrt{1-\delta_{k,\ell}}\ket{\mathcal{M}_{k,\ell}} + \sqrt{\delta_{k,\ell}} \ket{\mathcal{M}_{k,\ell}^\perp},
$$
using a number of queries to $U$ that scales as $O(\gamma^{-2} \log^2(1/\gamma) \log(\log(1/\gamma)/\delta))$ where $\ket{\mathcal{M}_{k,\ell}}\in \mathbb{C}^{n+a}$ has projection only on states that encode of $\mathbf{M}_{k,\ell}$ up to $\gamma$--error meaning that for any computational basis state ${\rm Tr} (\ketbra{y}{y}{\rm Tr_{junk}} (\ketbra{\mathcal{M}_{k,\ell}}{\mathcal{M}_{k,\ell}})) \ne 0$  if $|y-\mathbf{M}_{k,\ell}|> \gamma$.\label{cor:covbd}
\end{corollary}
\begin{proof}
The proof is a direct consequence of~\thm{PCAmain}.  The computation of the medians involve three steps.  First we need to compute the median for fixed $k$ and the median for fixed $\ell$ and then compute the median of the products of differences between these terms and the inner products.  From~\thm{PCAmain} the number of queries to the vectors needed to estimate this is in $O(\gamma^{-2} \log^2(1/\gamma) \log(\log(1/\gamma)/\delta))$ because the costs of both of these processes is additive.  Because the inner products are bounded above by $1$ the error from estimating the medians in the products of terms is $O(\gamma)$.  Thus the error from the computation of both medians sums to a value of $O(\gamma)$ if $O(\gamma^{-2} \log^2(1/\gamma) \log(\log(1/\gamma)/\delta))$ queries are used for each of the three procedures.
\end{proof}

Now that we have established the cost of implementing an approximate query to the matrix $\mathbf{M}$ the only step that remains is to diagonalize the matrix in order to be able to find the support of an input vector within a subspace.  For simplicity, we will assume in the following that the error tolerance for the median, $\epsilon$, required by the algorithms is constant. However, we will assume that the error tolerance required for the inner products, $\epsilon'$, is not a constant and will demand increased precision as the size of the problem grows.  For this reason, existing methods that show that the errors in the unitary implementation translate into error in the eigenphase do not necessarily hold (as they would for Trotter-based methods).   For this reason we need to explicitly demonstrate that phase estimation can be employed using this approach with finite error.

Before concluding with our main theorem on robust quantum PCA, there is one more technical lemma that needs to be proven.  We need to demonstrate that when using non-unitary simulation methods for $e^{-i\mathbf{M}}$ that the errors that arise from the truncated Taylor simulation method do not invalidate the phase estimation protocol.

\begin{lemma}
Let $\mathbf{M}\in \mathbb{C}^{2^n\times 2^n}$ and $\widetilde{\mathbf{M}}\in \mathbb{C}^{2^n\times 2^n}$ be $d$-sparse Hermitian matrices  and let $\mathbf{M}^{(j)}:j=1,\ldots, \kappa d^2$ be a set of one-sparse matrices such that each non-zero matrix element in $\mathbf{M}$ is assigned to precisely one $\mathbf{M}^{(j)}$ where $f(p,j)$ yields the non-zero matrix element in row $p$ of $\mathbf{M}^{(j)}$. If we then define $V$ to be the set of outcomes for which amplitude estimation estimates the desired matrix element of $\mathbf{Y}$ within error $\epsilon$ and \label{lem:Mbd}
\begin{itemize}
\item $F(j,m,p,x):= (-1)^{m(x M < m \|\mathbf{M}\|_{\max})} \ketbra{p}{f(p,j)}$.
\item $\frac{\mathbf{M}}{\|\mathbf{M}\|_{\max}}= \frac{1}{M}\sum_{j=1}^{\kappa d^2}\sum_{m=1}^M \sum_{p=1}^{2^n} F(j,m,p,\mathbf{M}_{p,f(p,j)})$,
\item $\frac{\widetilde{\mathbf{M}}}{\|\mathbf{M}\|_{\max}} =\frac{1}{M} \sum_{j=1}^{\kappa d^2} \sum_{m=1}^M \sum_{p=1}^{2^n} \left([1-\delta_{j,p}] \sum_{k\in V} {b_k} F(j,m,p,\mathbf{Y}_{p,f(p,j),k}) + {\delta_{j,p}} \sum_{k\not \in V} {c_k} F(j,m,p,\mathbf{Y}_{p,f(p,j),k})\right)$,
\end{itemize} 
where $0\le \delta_{j,p} \le \delta \le 1$, $\|c\|_1=\|b\|_1=1$ and for all $k\in V $, $|\mathbf{Y}_{p,f(p,j),k} - \mathbf{M}_{p,f(p,j)}|\le \eta$.  Under these assumptions $\|\mathbf{M} - \widetilde{\mathbf{M}}\|\in O\left(d[\|\mathbf{M}\|_{\max}\delta +\eta]\right)$.
\end{lemma}
\begin{proof}
We prove this lemma using the fact that if $A$ and $B$ are $d$--sparse matrices with the same sparsity pattern then $\|A-B\|\le (d+2) \|A-B\|_{\max}$.  To see this note from Vizing's theorem that at most $d+1$ colors are required to edge color a degree $d$ graph since the diagonal elements can be viewed as having self loops in the worst case scenario one additional color may be needed hence $d+2$ colors is an upper bound.  Since we can view any Hermitian matrix as a weighted adjacency matrix of a graph and since the adjacency matrix for any graph that has degree at most $1$ is one-sparse, Vizing's theorem implies that $A-B$ can be decomposed into at most $d+2$ one--sparse matrices since $A$ and $B$ are $d$--sparse and share the same sparsity  pattern.  

Next for any one sparse matrix $C$ such that $C_{x,y}\ne 0$ then the irreducible sub-matrix in the span of $\ket{x}$ and $\ket{y}$ is either of the form $\begin{bmatrix}0&C_{x,y}\\C^*_{x,y}&0\end{bmatrix}$ or $\begin{bmatrix}{C_{x,x}}\end{bmatrix}$ if $x\ne y$ and if $x=y$ respectively.  In both cases, exact diagonalization shows that the eigenvalues of $C$ within this span are $\pm |C_{x,y}|$.  It therefore follows from this observation, the triangle inequality and the fact  that Vizing's coloring uniquely assigns each matrix element (edge in the adjacency matrix) to a unique one-sparse matrix
\begin{equation}
\|A-B\| \le (d+2)\max_{j} \|A^{(j)} -B^{(j)}\| = (d+2) \max_{j} \|A^{(j)} -B^{(j)}\|_{\max}=(d+2)\|A-B\|_{\max}.\label{eq:ABbd}
\end{equation}

With~\eq{ABbd} in hand it is clear that in order to prove our result we only need to bound $\|\mathbf{M}^{(j)} - \widetilde{\mathbf{M}}^{(j)}\|_{\max}$, where $\widetilde{\mathbf{M}}^{(j)}$ is the the sub-matrix in the decomposition of $\widetilde{\mathbf{M}}$ that corresponds to ${\mathbf{M}}^{(j)}$.  This result corresponds to
\begin{align}
&\left\|\frac{1}{M} \sum_{m=1}^M\left([1-\delta_{j,p}] \sum_{k\in V} {b_k} F(j,m,p,\mathbf{Y}_{p,f(j,p),k}) + {\delta_{j,p}} \sum_{k\not \in V} {c_k} F(j,m,p,\mathbf{Y}_{p,f(j,p),k})\right)-F(j,m,p,\mathbf{M}_{p,f(p,j)}) \right\|_{\max}\nonumber\\
&\le \left\|\frac{1}{M} \sum_{m=1}^M\sum_{k\in V} {b_k} F(j,m,p,\mathbf{Y}_{p,f(j,p),k}) -F(j,m,p,\mathbf{M}_{p,f(p,j)})  \right\|_{\max}\nonumber\\
&\qquad + \delta \max_m\left\| \sum_{k\in V} {b_k} F(j,m,p,\mathbf{Y}_{p,f(j,p),k}) +  \sum_{k\not \in V} {c_k} F(j,m,p,\mathbf{Y}_{p,f(j,p),k})\right\|_{\max}\nonumber\\
&\le \left\|\frac{1}{M} \sum_{m=1}^M\sum_{k\in V} {b_k} F(j,m,p,\mathbf{Y}_{p,f(j,p),k}) -F(j,m,p,\mathbf{M}_{p,f(p,j)})  \right\|_{\max}+2\delta\nonumber\\
&\le \frac{1}{M} \sum_{m=1}^M\sum_{k\in V} {b_k} \Delta_{j,p} +2\delta=\frac{1}{M} \sum_{m=1}^M \Delta_{j,p} +2\delta:=\bar{\Delta}+2\delta.\label{eq:2deltas}
\end{align}

Next we need to bound $\bar{\Delta}$.  We achieve this by bounding the propagation of the errors in each of the $\mathbf{M}^{(j)}_{p,f(p,j)}$ calculations to errors in $ (-1)^{m(\mathbf{M}^{(j)}_{p,f(p,j)} M < m \|\mathbf{M}\|_{\max})}$.  This is not direct because of the discrete nature of the phases used which means that even if the error in the calculation is non-zero the error in the coefficients for many of the unitaries in the Hamiltonian decomposition will be zero.  However, when errors occur $\Delta_{j,p}=2$.  Let us assume that $\mathbf{M}_{p,f(p,j)}$ is computed within error $\eta$.  Thus for fixed $p$ and $j$ the maximum probability over $m$ of such an error occuring is simply the quotient of the number of values of $m$ such that $(\mathbf{M}^{(j)}_{p,f(p,j)} M < m \|\mathbf{M}\|_{\max})$ is incorrect and $M$.  Let us define this uniform bound on the failure probability to be $P$ and define the error in the calculation of $\mathbf{M}^{(j)}_{p,j}$ to be $\eta$
\begin{equation}
P(fail|p,j) \le P := \frac{2}{M}\left\lceil\frac{\eta M}{\|\mathbf{M}\|_{\max}} \right \rceil \in O\left(\frac{\eta}{\|\mathbf{M}\|_{\max}} \right).
\end{equation}

Fortunately, the net error scales with the expectation value of $\Delta_{j,p}\le 2$ obeys
\begin{equation}
\bar{\Delta}  \le 2 P(fail|p,j) \in O\left(\frac{\eta}{\|\mathbf{M}\|_{\max}} \right).\label{eq:2deltas2}
\end{equation}

The result then follows from combining~\eq{ABbd},~\eq{2deltas} and~\eq{2deltas2}.
\end{proof}

It may be natural to assume that when you use oracles that are inherently probabilistic that the resulting quantum simulation of the Hamiltonians described by such oracles would inherently be non-deterministic also.  We show below that this need not be true.  In particular, the Hamiltonian simulated by the linear combination is in fact the average Hamiltonian over the ensemble of outputs.  We prove this below.

\begin{lemma}
Under the assumptions of \lem{Mbd}, the number of queries to the approximate oracle $V_{\mathbf{M}}$ and $f$ needed to implement a non-unitary operation $Q$ with failure probability at most $\epsilon$ such that $\|Q-e^{-i\widetilde{\mathbf{M}}}\|\le \epsilon$ and  $\|\mathbf{M}- \widetilde{\mathbf{M}}\|\le (\epsilon + d\gamma)$ with probability of failure $O(\epsilon)$  is in $O(d^2 \log(d^2/\epsilon)/\log\log(d^2/\epsilon)).$\label{lem:sim1}
\end{lemma}
\begin{proof}
The simulation method used is, in essence, the trucated Taylor series simulation method in~\cite{berry2015simulating}.  Since that result assumes that the Hamiltonian oracles yield the precise matrix elements, rather than probabilistic approximations to them, we cannot directly use the results without modification.

The truncated Taylor series approximation to $e^{-i\mathbf{M} t}$ takes the form
\begin{equation}
\sum_{q=0}^K \frac{(-i)^q}{q!} [\mathbf{M}t]^q = \sum_{q=0}^K \beta_q [\mathbf{M}t]^q = \sum_{q=0}^K \beta_q \prod_{\ell =1}^q \left(\sum_{j=1}^{\kappa d^2} \mathbf{M}^{(j)}t\right).
\end{equation}
Now consider the  following oracle
\begin{equation}
O_{\mathbf{M}} \ket{j,0,p} = \sqrt{1-\delta_{j,p}} \left(\sum_{k\in V} \sqrt{b_k} \ket{\mathbf{Y}_{p,f(p,j),k}}+ \sqrt{\delta_{j,p}} \sum_{k\not\in V} \sqrt{c_k} \ket{\mathbf{Y}_{p,f(p,j),k}}\right)\ket{p},
\end{equation}
where $\mathbf{Y}$ and $f(p,j)$ are defined in~\lem{Mbd}.  This oracle can be constructed using $O(1)$ fundamental calls to the matrix element oracle for $\mathbf{M}$ and $f$.

Next using the definitions presented in \lem{Mbd} we have that $\mathbf{M}$ can be written (recalling that $\|\mathbf{M}\|_{\max}\le 1$) is
\begin{equation}
 {\mathbf{M}}= \frac{\|\mathbf{M}\|_{\max}}{M}\sum_{j=1}^{\kappa d^2}\sum_{m=1}^M \sum_{p=1}^{2^n} F(j,m,p,\mathbf{M}_{p,f(p,j)})+O\left(\frac{d^2}{M} \right).
\end{equation}
Now let  $y$ be a boolean function that evaluates to
\begin{equation}
y_{jmpk}=\mathbf{Y}_{p,f(p,j),k} M > m.\label{eq:phase}
\end{equation}
Then a quantum circuit that implements $F$ can be executed using $O(1)$ queries to the matrix elements of $\mathbf{M}$ and a polynomial-sized arithmetic circuit to compute~\eq{phase} and kick the result back into the phase.  Specifically this allows us to perform a transformation  that we denote PrepareW$^{(1)}$ 
\begin{align}
&\ket{0}\ket{\psi} =\ket{0} \sum_{p} \sqrt{\psi_p} \ket{p} \mapsto \nonumber\\
&\frac{1}{\mathcal{N}} \left(\ket{0}\ket{\psi} + \frac{1}{\sqrt{M}} \sum_{j,m,p} \sqrt{\psi_p}\ket{j.m}  \left(\sqrt{1-\delta_{j,p}}\sum_{k\in V} \sqrt{b_k} \ket{\mathbf{Y}_{p,f(p,j),k}}\ket{y_{jmpk}}+ \sqrt{\delta_{j,p}} \sum_{k\not\in V} \sqrt{c_k} \ket{\mathbf{Y}_{p,f(p,j),k}}\ket{y_{jmpk}}\right)\ket{p} \right),
\end{align}
where $\mathcal{N}$ is a normalizing constant.  This transformation takes the form of the PrepareW oracle used in the LCU lemma~\cite{berry2014exponential}.  By invoking PrepareW$^{(1)}$, $F$ as well as an appropriate phase correction to add a factor of $-i$ to the linear term and the inverse of PrepareW the LCU lemma states that, upon success, we apply
\begin{equation}
\ket{0}\ket{\psi} \mapsto \frac{\left(\openone -i \widetilde{\mathbf{M}} \right)\ket{\psi}}{\|\left(\openone -i \widetilde{\mathbf{M}} \right)\ket{\psi}\|} + O(d^2\|\mathbf{M}\|_{\max}/M),
\end{equation}
where $\tilde{\mathbf{M}}$ is given by~\lem{Mbd}.  

This illustrative example shows why when non-deterministic oracles are used within an LCU Hamiltonian simulation
Now let $f(p,\vec{j})$ be defined to be the oracle that yields the column index of the non-zero element in the $p^{\rm th}$ row of the product of several one-sparse matrices whose $j$ indices are stored in the vector $\vec{j}$.  This is given recursively via 
$$f(p,\vec{j}) = f(f(p,\vec{j}\setminus \vec{j}_{{\rm dim} (\vec{j})}),\vec{j}_{{\rm dim} (\vec{j})}).$$
For simplicity we also define $f(p,\emptyset) = p$.
Then for dim$(\vec{j})=J$,  if we call the approximate oracle $J$ times to compute the matrix element corresponding to a fixed $\vec{j}$ and $p$ we obtain 
\begin{equation}
\ket{\vec{j}}\bigotimes_{q=1}^J\left(\sqrt{1-\delta_{j_q,p}}\sum_{k\in V} \sqrt{b_k} \ket{\mathbf{Y}_{f(p,[j_1,\ldots,j_{q-1}]),f(p,[j_1,\ldots,j_q]),k}}+ \sqrt{\delta_{j_q,p}} \sum_{k\not\in V} \sqrt{c_k} \ket{\mathbf{Y}_{f(p,[j_1,\ldots,j_{q-1}]),f(p,[j_1,\ldots,j_q]),k}}\right)\ket{p}\label{eq:expansion}
\end{equation}
Now let us define the analog of $y$ for vector valued inputs $j$ and $k$ where $\vec{k}$ be the vector of $k$ values that arises from expanding~\eq{expansion} to be
\begin{equation}
y_{\vec{j}mp\vec{k}}:= M\prod_{q=1}^J \mathbf{Y}_{f(p,[j_1,\ldots,j_{q-1}]),f(p,[j_1,\ldots,j_q]),k_q}/J!\le  m.
\end{equation}
It then follows that by acting on a superposition over $j$ vectors,  we can prepare a state using $J$ queries to the approximate oracle that is, up to isometries,
\begin{equation}
\frac{1}{\sqrt{M}} \sum_{\vec{j},m,p}\!\!\left( \sqrt{\psi_p}\ket{\vec{j}.m,p}\!  \left(\prod_{q=1}^J\sqrt{1-\delta_{j_q,p}}\sum_{\vec{k}\in V^k} \sqrt{b_k} \ket{y_{\vec{j}mp\vec{k}}}+ \sum_{\ell=1}^J\sqrt{\delta_{j_\ell,p}}\prod_{q\ne \ell} \sqrt{1-\delta_{j_q,p}} \sum_{\vec{k} \setminus \vec{k}_q \in V, \vec{k}_q \not \in V} \sqrt{c_k} \ket{y_{\vec{j}mp\vec{k}}}+\cdots\right)\!\! \right)\label{eq:generalquery}
\end{equation}
Let us define this process to be PrepareW$^{(J)}$.  

From the LCU lemma and the binomial theorem, an invocation of PrepareW$^{(J)}$, $J$ iterations of the $F$ oracle (and some additional phase rotations needed to set $(-i)^J$) and the inverse of PrepareW$^{(J)}$, conditioned on a measurement of $0$ enacts
\begin{equation}
\ket{0} \ket{\psi} \mapsto (-i\widetilde{\mathbf{M}})^J \ket{\psi} + O\left({[d^{2}\|\mathbf{M}\|_{\max}]^J}\left[ \frac{1}{M} +\delta^2\right]\right).\label{eq:higherord}
\end{equation}
If we take $\delta \in O(1/M)$ then the LCU lemma and~\eq{higherord} imply that for all $J>0$ we can, through the use of an auxillary register that encodes the $J$ different values in unary, prepare the state
\begin{equation}
\ket{0}\ket{\psi} \mapsto  \frac{\ket{0}\sum_{q=0}^J (-i\widetilde{\mathbf{M}})^{q}/q!\ket{\psi}}{\|\sum_{q=0}^J (-i\widetilde{\mathbf{M}})^{q}/q!\ket{\psi}\|}+ O\left(\frac{d^{2J}}{M} \right).\label{eq:postsel}
\end{equation}
Where above we use the assumption that $\|\mathbf{M}\|_{\max} \le 1$.  This takes the same form as the simulation error for linear combination of unitary expansions that have no error.  Thus at this point we can follow the remainder of the derivation in~\cite{berry2015simulating}.

Robust oblivious amplitude amplification is used to boost the probability to nearly $100\%$, however, there are limits on the effectiveness of this technique based on the non-unitarity of the underlying operations.  By taking $J\in O(\log(d^2/\epsilon)/\log\log(d^2/\epsilon))$ we can guarantee that the error in the truncated Taylor series simulation obeys for a given number of segments $r$~\cite{berry2015simulating} $$\max_{\ket{\psi}}\left\|e^{-i\widetilde{\mathbf{M}}/r}\ket{\psi} - \frac{\sum_{q=0}^J (-i\widetilde{\mathbf{M}}/r)^{q}/q!}{\|\sum_{q=0}^J (-i\widetilde{\mathbf{M}}/r)^{q}/q!\ket{\psi}\|}\ket{\psi}\right\|\le \frac{\epsilon}{r}.$$
If we let $Q$ be the non-unitary approximation that arises from potential failures in the LCU method leading the errors in the expansion, then we obtain
\begin{equation}
\left\|e^{-i\widetilde{\mathbf{M}}/r}\ket{\psi} - Q\ket{\psi}\right\|\in O\left(\frac{\epsilon}{r} +\frac{d^{2J}}{M}  \right)\label{eq:}
\end{equation}
Thus if we use the form of robust oblivious amplitude amplification discussed in~\cite{berry2015simulating} and take $M\in \Theta(d^{2J}r/\epsilon)$ then the total non-unitary error of implementing the $r\in \Theta(d^2)$ segments of the simulation~\cite{berry2015simulating} is in $O(\epsilon)$.  The overall query complexity is $O(Jr)$ which $O(d^2\log(d^2/\epsilon)/\log\log(d^2/\epsilon))$ if simulation error $\epsilon$ is desired for a simulation of $e^{-i\widetilde{\mathbf{M}}}$.  

Similarly, we have from~\lem{Mbd}, the triangle inequality and box 4.1 from Nielsen and Chuang~\cite{NC00} that $\|e^{-i\mathbf{M}} - e^{-i\widetilde{\mathbf{M}}}\|\in O( d[\delta +\gamma])$.  Thus we require that $\delta \in O(\epsilon/d)$ to ensure that the error is $O(d\gamma +\epsilon)$.
It then follows that the error in the approximation can be made at most $\epsilon$ for $M \in O(d^{2J}/\epsilon)$,
which implies that the circuit used to prepare the $M$ register (which is a uniform superposition) is of size at most $O(J\log(d)+\log(1/\epsilon))$ qubits.   Thus the constraint that $\delta \in O(1/M)$ directly implies $\delta \in O(\epsilon/d)$ and furthermore the space and time required to prepare such a register is polynomial.
\end{proof}

This shows that we can perform a simulation using LCU methods using a non-deterministic oracle for the matrix elements of the underlying Hamiltonian.  Note the above reasoning explicitly holds for the case of simulating $e^{-i Ht}$ for sparse Hermitian $H$ by substituting $\mathbf{M} \rightarrow H t$.

\begin{theorem}\label{thm:mainPCA}
Under the assumptions of~\lem{Mbd} and given that the minimum eigenvalue gap of $\mathbf{M}$ within the support of the input state obeys $\lambda \ge \epsilon/4>0$ .  We then have the following:
\begin{enumerate}
\item The number of queries to the data set needed to perform a non-unitary operation $Q$ such that $\|e^{-i\widetilde{\mathbf{M}}}-Q\|\in O(\epsilon)$ for $\|\mathbf{M}-\widetilde{\mathbf{M}}\|\in O( \epsilon)$ and failure probability in $O(\epsilon)$ is in ${O}([d^4/\epsilon^2] \log^3(d^2/\epsilon)) $
\item The number of queries to the data set needed to ensure that the projection of an arbitrary test vector onto the principal components of  $\widetilde{\mathbf{M}}$ within a misclassification error probability of $\Lambda\in (0,4]$ is in $O([d^4/\Lambda^2 \lambda^2]\log^3(d^2/(\Lambda \lambda)))$. 
\end{enumerate}
\end{theorem}
\begin{proof}
From~\cor{covbd} we have that the number of queries needed to the training data and an oracle that gives the sparsity structure of the robust covariance matrix is in $O(\gamma^{-2} \log^2(1/\gamma) \log(\log(1/\gamma)/\delta))$.  From~\lem{sim1} the number of such queries is in $O(d^2\log(d^2/\epsilon)/\log\log(d^2/\epsilon))$, since each query to $f$ has query complexity $1$ which is asymypotically subdominant to the query complexity required to simulate $V_{\mathbf{M}}$.  Thus the total number of queries needed to the training data is for $\gamma\in O(\epsilon/d)$ 
\begin{equation}
O((d^4/\epsilon^2) \log(d^2/\epsilon)\log^2(d/\epsilon)\log(d\log(d/\epsilon)/\epsilon)/\log\log(d^2/\epsilon))\subseteq O((d^4/\epsilon^2)\log^3(d^2/\epsilon)).\label{eq:simbd}
\end{equation}  

Next we need to consider the effect of such errors in $\mathbf{M}$ on the eigenvalues of the robust covariance matrix.  In order to understand the impact of such errors we must ensure that the errors incurred in the principal components of the approximated robust density matrix $\widetilde{\mathbf{M}}$ are small.  In general, if the perturbations to the robust covariance matrix are large relative to the eigenvalue gaps then we can expect that the impact will be substantial for both the eigenvalue and eigenvector in question.  To this end, we will have to bound the impact that these perturbations will have on the result.

Let $\|\mathbf{M} -\widetilde{\mathbf{M}}\|\le \sigma$ then there exists an a matrix valued function with norm at most $1$, $\Delta(\sigma)$, such that $\mathbf{M} = \widetilde{\mathbf{M}}+\sigma\Delta(\sigma)$.  Now if $\widetilde{\mathbf{M}}$ is an analytic matrix-valued function of $\sigma$ then so is $\Delta(\sigma)$.   In this context, the path that we choose to perturb $\mathbf{M}$ along to reach $\widetilde{\mathbf{M}}$ is arbitrary.  We therefore choose it to be analytic.  Thus we can assume without loss of generality that $\Delta(\sigma)=\Delta + O(\sigma)$ for constant operator $\Delta$ such that $\|\Delta\|\le 1$.  The latter point follows by contradiction, if $\|\Delta\|>1$ then the requirement that $\|\Delta(\sigma)\|\le 1$ will fail to hold in some $\sigma$-neighborhood about $0$.

Now we will investigate the role that such errors have in the shifts in the eigenvalues and eigenvectors of $\mathbf{M}$.
Since $\Lambda>0$ we can apply non-degenerate perturbation theory to estimate the shift in the eigenvalues and eigenvectors.  If we define $E_n$ to be the $n^{\rm th}$ eigenvalue of $\mathbf{M}$ and $\ket{E_n}$ to be the corresponding eigenvector, and $E_n'$ to be the corresponding perturbed eigenvalue we have from non-degenerate perturbation theory that~\cite{ballentine2014quantum}.
\begin{equation}
E_n'(\sigma) = E_n + \bra{E_n} \Delta \ket{E_n}\sigma + O(\sigma^2).
\end{equation}
This implies that the directional derivative of the eigenvalue in the direction of the perturbation is
\begin{equation}
\partial_{\sigma} E_n(0)  := \lim_{\sigma\rightarrow 0} \frac{E_n'(\sigma)-E_n}{\sigma}= \bra{E_n} \Delta \ket{E_n}.\label{eq:eigderiv0}
\end{equation}
Note that because the result is insensitive to the initial value, the same argument can be used to bound the derivative for any value of $\sigma$.
Similarly, the derivative of the projection of the eigenstate onto any fixed vector $\ket{x}$ is
\begin{equation}
\partial_{\sigma} \braket{x}{E_n'(0)} = \sum_{p\ne n} \frac{\braket{x}{E_p}\bra{E_p} \Delta \ket{E_n}}{E_p-E_n}.
\end{equation}

Now let us define the space to be the direct sum of three disjoint spaces eigenspaces of $\mathbf{M}$ defined by the projectors $P_+$, $P_-$ and $P_?$, where $\openone =P_+ + P_- + P_?$.  Furthermore, assume that for any $\ket{\psi}$ in the $+1$ eigenspace of $P_+$ and $\ket{\phi}$ in the $+1$ eigenspace of $P_-$, $|\bra{\psi} \mathbf{M}\ket{\psi} - \bra{\phi} \mathbf{M} \ket{\phi}|\ge \lambda$.  Similarly let $P_+(\epsilon)$ be the projector onto the perturbed $+1$ eigenspace.  Under the assumption that $\mathbf{M}\in \mathbb{R}^{N\times N}$ we have that the eigenvectors can always be chosen to be real valued and hence for any $\ket{\psi}$,
\begin{align}
\partial_\sigma \bra{\psi}P_+(\sigma)\ket{\psi} &= 2\sum_{p\ne n} \sum_{n\in P_+} \braket{\psi}{E_n}\braket{E_p}{\psi}\frac{\bra{E_p} \Delta \ket{E_n}}{E_p-E_n}=2\sum_{p\ne n} \sum_{n\in P_+} \braket{\psi}{E_n}\braket{E_p}{\psi}\frac{\bra{E_p} \Delta \ket{E_n}}{E_p-E_n},\nonumber\\
&=2\sum_{p\ne n} \sum_{n\in P_+} \bra{\psi} (\ketbra{E_p}{E_p}) \Delta \left(\frac{\ketbra{E_n}{E_n}}{E_p-E_n}\right)\ket{\psi}\le 2\left\|\sum_{p\ne n} \sum_{n\in P_+}  (\ketbra{E_p}{E_p}) \Delta \left(\frac{\ketbra{E_n}{E_n}}{E_p-E_n}\right) \right\|\nonumber\\
&=2\left\|\sum_{p}   (\ketbra{E_p}{E_p}) \Delta \left(\sum_{n\in P_+} [1-\delta_{n,p}]\frac{\ketbra{E_n}{E_n}}{E_p-E_n}\right) \right\|\nonumber\\
&:=2\left\|\sum_p\ketbra{E_p}{E_p} \Delta O_p\right\|\le 2\|\openone\|\|\Delta\|\|O_p\|\le \frac{2}{\lambda}.
\end{align}
where $n\in P_+$ implies $P_+ \ket{E_n} = \ket{E_n}$ and the last line follows from $|E_p-E_n| \ge \lambda$.
Thus if the eigenvalues are computed within an error of at most $\lambda/4$ then it follows that the maximum error that we can observe in the derivative is $4/\lambda$.  Thus
\begin{equation}
|\bra{\phi} P_+(\epsilon) - P_+ \ket{\phi}| \le \int_0^\epsilon \max_{\psi}|\partial_\sigma \bra{\psi}P_+(\sigma)\ket{\psi} |\le \frac{4\epsilon}{\lambda}. 
\end{equation}
This shows that provided that there is a gap between the two subspaces for good and bad examples the probability of misclassifying the vector based on its support in either class can be made $\Lambda$ by choosing $\epsilon \in \Theta(\Lambda\lambda)$.

\end{proof}

\thm{mainPCA} shows us that the query complexity of simulating the time evolution operator within sufficient error tolerance to ensure that the perturbation to the principal component vectors is negligible is independent of the dimension of the problem.  In contrast, classical methods typically require a number of queries that scales polynomially with the dimension.  The final step needed is to consider the cost of phase estimation.  We restate this formally below,

\finalPCA*
\begin{proof}
The result trivially follows from~\thm{mainPCA} and bounds on phase estimation that show that the cost of learning eigenvalues within error $\alpha$ and probability of failure $\beta$ is in $O(1/[\alpha \beta])$~\cite{NC00}.  In this context $\alpha\in \Theta(\Lambda)$ and $\beta\in \Theta(\lambda)$ and thus the result simply follows by multiplication.

The final part of the proof is to show that the error in the robust PCA matrix, with respect to the $2$-norm, is bounded above by $\alpha L (d+2)$ under these assumptions.  The robust PCA matrix we consider is of the form $\mathbf{M}_{k,\ell}= {\rm median}([e_k^Tx_j -{\rm median}(e_k^T x_j)][e_\ell^Tx_j -{\rm median}(e_\ell^T x_j)])$ from~\defn{N_v}.  We then have from \lem{perturb} and the assumptions that $\alpha < 1/2 < 1$ and $L\alpha\le 1$ that
\begin{align}
\mathbf{M'}_{k,\ell} &\le {\rm median}([e_k^Tx_j -{\rm median}(e_k^T x_j)+\alpha L][e_\ell^Tx_j -{\rm median}(e_\ell^T x_j)+\alpha L]),\nonumber\\
&\le {\rm median}([e_k^Tx_j -{\rm median}(e_k^T x_j)][e_\ell^Tx_j -{\rm median}(e_\ell^T x_j)]+ 2\max_{k,j}(e_k^Tx_j -{\rm median}(e_k^T x_j))\alpha L +(\alpha L)^2)\nonumber\\
&\le \mathbf{M}_{k,\ell}+ 4(e_k^T x_j))\alpha L +(\alpha L)^2)\le \mathbf{M}_{k,\ell}+ 5(e_k^T x_j))\alpha L.\label{eq:medianub}
\end{align}
Then by following the same argument used in~\eq{medianub} we arrive at a similar lower bound which shows that $\|\mathbf{M}- \mathbf{M'}\|_{\max} \le 5\alpha L$.  By assumption $\mathbf{M}$ is $d$--sparse, ergo it follows from~\eq{ABbd} that
\begin{equation}
\|\mathbf{M} -\mathbf{M'}\|_2 \le 5\alpha L (d+2).
\end{equation}
It therefore follows that the adversary's maximum impact on the robust PCA matrix is limitted, even in the worst case scenario, provided that they control a small fraction of the training data.  This vindicates that our robust quantum solution is still stable given~\threatM{1} for $\alpha>0$.
\end{proof}


\end{document}